\documentclass[journal]{IEEEtran}
\usepackage{amsmath,amsfonts}
\usepackage{algorithmic}
\usepackage{algorithm}
\usepackage{makecell}
\usepackage[caption=false,font=footnotesize]{subfig}
\usepackage[font=footnotesize]{caption}
\usepackage{graphicx}
\usepackage{pifont}
\usepackage{booktabs}
\usepackage{amsthm,amssymb}
\usepackage{microtype}
\usepackage{multirow}
\usepackage{tikz}
\usepackage{stfloats}
\usepackage[normalem]{ulem}
\usepackage{cite}
\usepackage{palatino}
\usepackage[hidelinks]{hyperref}
\usepackage[capitalize]{cleveref}

\newtheorem{theorem}{Theorem}
\newtheorem{proposition}{Proposition}

\def\A{\mathbf{A}}

\def\Re{\mathbb{R}}
\newcommand{\cD}{\mathcal{D}}
\newcommand{\cO}{\mathcal{O}}
\newcommand{\cS}{\mathcal{S}}
\newcommand{\cN}{\mathcal{N}}
\newcommand{\cI}{\mathcal{I}}
\newcommand{\cT}{\mathcal{T}}
\newcommand{\Tpgd}{\mathcal{T}_{\gamma}}
\newcommand{\Thqs}{\mathcal{T}_{\mu}}
\newcommand{\Tadmm}{\mathcal{T}_{\alpha}}

\newcommand{\xmark}{\textcolor{black}{\ding{55}}}%

\def\p{\boldsymbol{p}}
\def\x{\boldsymbol{x}}
\def\y{\boldsymbol{y}}

\begin{document}

\title{Viscosity Stabilized Plug-and-Play Reconstruction}

\author{Arghya~Sinha,
        Trishit Mukherjee,
        and~Kunal~N.~Chaudhury

\thanks{A.~Sinha, T.~Mukherjee and K.~N.~Chaudhury are with the Indian Institute of Science, Bengaluru: 560012, India. Correspondence: arghyasinha@iisc.ac.in. A. Sinha was supported by the PMRF fellowship TF/PMRF-22-5534 from the Government of India. K.~N.~Chaudhury was supported by grant STR/2021/000011 from ANRF, Government of India.}
}

\markboth{Sinha \MakeLowercase{\textit{et al.}}: ViSTA-PnP}%
{Sinha \MakeLowercase{\textit{et al.}}: ViSTA-PnP}

\maketitle
\begin{abstract}
The plug-and-play (PnP) method uses a deep denoiser within a proximal algorithm for model-based image reconstruction (IR). Unlike end-to-end IR, PnP allows the same pretrained denoiser to be used across different imaging tasks, without the need for retraining. However, black-box networks can make the iterative process in PnP unstable. A common issue observed across architectures like CNNs, diffusion models, and transformers is that the visual quality and PSNR often improve initially but then degrade in later iterations. Previous attempts to ensure stability usually impose restrictive constraints on the denoiser. However, standard denoisers, which are freely trained for single-step noise removal, need not satisfy such constraints. We propose a simple data-driven stabilization mechanism that adaptively averages the potentially unstable PnP operator with a contractive IR operator.  This acts as a form of viscosity regularization, where the contractive component progressively dampens updates in later iterations, helping to suppress oscillations and prevent divergence. We validate the effectiveness of our stabilization mechanism across different proximal algorithms, denoising architectures, and imaging tasks.
\end{abstract}


\section{Introduction}
\label{sec:intro}

\IEEEPARstart{T}{he} problem of recovering an image from noisy linear measurements comes up in applications such as deblurring, superresolution, magnetic resonance imaging, and tomography~\cite{bouman2022foundations}. The measurement process for such applications is modeled as
\begin{equation}
\label{eq:fm}
\y =\A \bar{\x}+ \boldsymbol{\epsilon},
\end{equation}
where $\bar{\x} \in \Re^n$ and $\y\in \Re^m$ represent the ground-truth and observed images, $\A \in \Re^{m \times n}$ is the forward operator, and $\boldsymbol{\epsilon} \in \Re^m$ represents noise. For example, in motion deblurring, $\A$ represents a blur operator, while in superresolution, it represents lowpass filtering followed by downsampling. This is a classical inverse problem with a rich literature, including variational methods~\cite{rudin1992nonlinear,geman1986markov}, trained networks~\cite{jin2017deep,guo_mambair_2025,liang_swinir_2021,zamir2022restormer}, denoiser-driven regularization~\cite{venkatakrishnan2013pnp,sreehari2016plug,romano2017little,zhang2021plug,hurault2022gradient},
and diffusion methods \cite{chung_diffusion_2022,kawar_denoising_2022,zhu_denoising_2023,lugmayr_repaint_2022}.

\subsection{PnP Algorithm}

Lately, denoisers have become popular tools for image generation and regularization~\cite{elad_image_2023}. The focus of this work is the plug-and-play (PnP) method, which leverages powerful denoisers as image regularizers within classical iterative algorithms~\cite{zhang2021plug,yuan2020plug,wei2020tuning,kamilov2023plug,ahmad2020plug,venkatakrishnan2013pnp}. The PnP method builds on the standard variational framework for solving~\eqref{eq:fm}, given by
\begin{equation}
\label{eq:fg}
\underset{\x \in \Re^n}{\min} \  f(\x) + g(\x), \quad f(\x) = \frac{1}{2} \lVert \A\x-\y \rVert^2,
\end{equation}
where $f: \Re^n \to \Re$ is a model-based loss and $g: \Re^n \to [0, \infty]$ is a convex regularizer that incorporates prior knowledge about the ground-truth image. 

The composite optimization problem~\eqref{eq:fg} can be solved iteratively using proximal algorithms~\cite{parikh2014proximal,bauschke2019convex}. A widely used algorithm is proximal gradient descent (PGD), which proceeds as follows:
\begin{equation}
\label{eq:pgd}
\begin{cases}
\x_0 \in \Re^n, \\
\x_{k+1} = \mathrm{prox}_{\gamma g} \big(\x_k - \gamma \nabla \! f(\x_k) \big) \qquad (k \geqslant 0),
\end{cases}
\end{equation}
where $\gamma > 0$ is the step size and $\mathrm{prox}_{\gamma g}$ denotes the proximal operator of the function $\gamma g$. 

The latter effectively acts as a denoiser or smoothing operator in~\eqref{eq:pgd}, helping to reduce artifacts introduced by the gradient step. Much research has focused on designing handcrafted image regularizers that admit a powerful proximal operator~\cite{bouman2022foundations}.

The original idea in PnP~\cite{venkatakrishnan2013pnp,sreehari2016plug} was to replace the proximal operator in~\eqref{eq:pgd} with an off-the-shelf a denoiser, i.e, the update in~\eqref{eq:pgd} is performed as
\begin{equation}
\label{eq:pnppgd}
\x_{k+1} = \cD \big(\x_k- \gamma \nabla \!f(\x_k)\big),
\end{equation}
where $\cD: \Re^n \to \Re^n$ is some denoising operator.
Proximal algorithms such as Alternating Direction Method of Multipliers~(ADMM) and Half-Quadratic Splitting~(HQS)~\cite{beck2017book,geman1995nonlinear} can also be used in place of PGD.

The key advantage of PnP over end-to-end reconstruction~\cite{zamir2022restormer,guo_mambair_2025,liang_swinir_2021} is that the forward model~\eqref{eq:fm} is used only during inference, and not during training. This removes the need for task-specific training, allowing the same pretrained denoiser to be used for different imaging tasks. However, this flexibility comes at a cost. Unlike end-to-end networks, which apply the model just once, PnP methods repeatedly invoke the denoiser within an iterative loop. This repeated use of a black-box network can lead to unstable behavior.

There has been considerable research on designing denoisers that ensure convergence of PnP algorithms~\cite{pesquet_learning_2021,hertrich_convolutional_2021,cohen_regularization_2021,hurault2022gradient,hurault_proximal_2022,nair2024averaged,goujon2023neural,goujon_learning_2024,pourya2025dealing}. The technical challenge is to train neural networks that not only offer theoretical convergence guarantees but also deliver strong denoising performance. Striking this balance often involves introducing structural constraints into the network through carefully designed architectures or tailored parameterizations.  However, this can limit the regularization capacity of the denoiser, which in turn may affect the reconstruction quality.~\cite{pesquet_learning_2021,nair2024averaged}.

\begin{figure*}[thbp]
    \centering
    \includegraphics[width=1\textwidth]{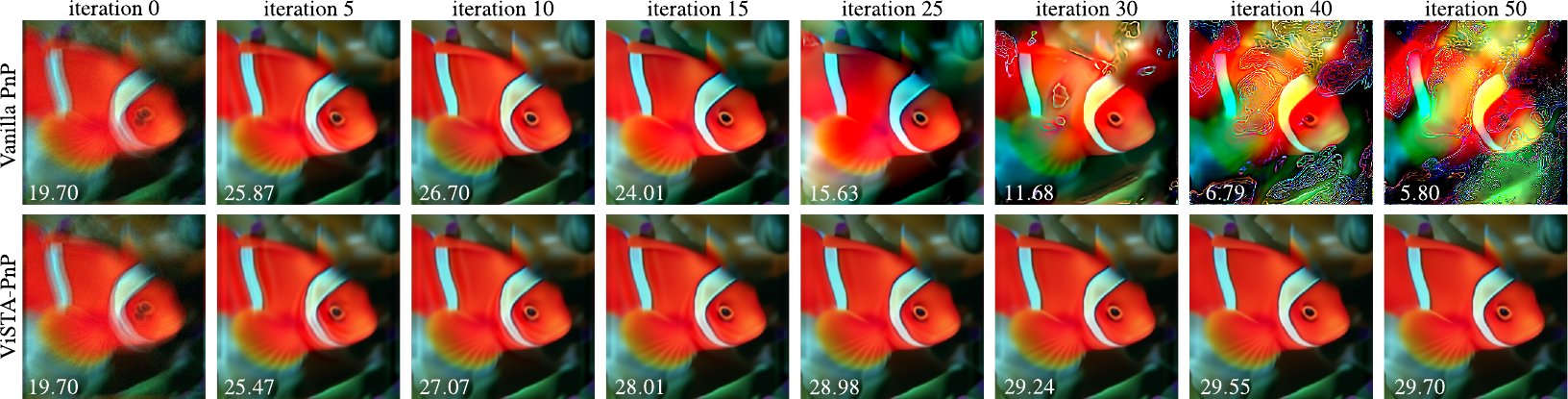}
    \caption{An example showing the instability of Vanilla-PnP (\textbf{top}) for motion deblurring on the \emph{fish} image~\cite{general100}. The base algorithm in this example is PGD~\eqref{eq:Tpgd} with step size $\gamma=2.2$, and the denoiser is MambaIRv1~\cite{guo_mambair_2025}. Notice how the PSNR steadily improves for the first $15$ iterations, following which the process degenerates and we see a propagation of artifacts across the image. The proposed stabilization (\textbf{bottom}) prevents the breakdown, while preserving the peak PSNR achieved just before instability.}
    \label{fig:frontimage}
\end{figure*}

In this work, we explore what happens when a pretrained denoiser is used inside a PnP framework. Such denoisers are usually trained for single-step denoising and do not have the structural properties needed to guarantee convergence of the sequence \(\{\x_k\}\) produced by~\eqref{eq:pnppgd}. As a result, their behavior can be unpredictable. In the context of image reconstruction, we observe a consistent failure pattern that forms the basis of this study. Specifically, when popular denoisers are plugged into PGD, HQS, or ADMM, the reconstruction quality improves steadily during the early iterations, but then suddenly starts to degrade after a certain point (see~\cref{fig:frontimage}, first row). As illustrated in~\cref{fig:deep_den_instable}, this behavior is seen across various architectures, including CNNs, diffusion models, and transformers.

\subsection{Related Work}

There has been a lot of work on designing convergent PnP models using deep denoisers. Broadly, two main approaches have emerged. In the first approach, a neural network is used to define an explicit regularizer, which is then optimized together with the model-based loss in an iterative manner~\cite{romano2017little,cohen2021has,reehorst2018regularization,cohen_regularization_2021,hurault2022gradient,hurault_proximal_2022,tan2024provably,goujon2023neural,goujon_learning_2024,pourya2025dealing}. Typically, the network is optimized alongside the proximal algorithm to ensure they work well together. For example,  CNN-based Laplacian regularizers are used in \cite{hurault2022gradient,cohen_regularization_2021}, while \cite{goujon2023neural} builds on convex ridge functions. These methods often follow the structure of PGD, where the gradient step of the regularizer is trained to act like a denoiser. Similarly, \cite{pourya2025dealing,goujon_learning_2024} use quadratic or weakly convex regularizers and train denoisers tailored to their specific iterative schemes. However, the denoiser is closely tied to the base algorithm, and it cannot simply be swapped with another pretrained model while still guaranteeing convergence. 

In the second approach, PnP is viewed as a nonlinear dynamical system~\cite{ryu2019plug,zhang2021plug}, and its convergence is studied using ideas from fixed-point theory~\cite{bauschke2019convex}. One way to ensure convergence is to use classical pseudolinear denoisers~\cite{ACK2023-contractivity,arghya-tsp}. Another strategy is to train a parametric family of denoisers~\cite{pesquet_learning_2021,hertrich_convolutional_2021} that are specifically designed to satisfy mathematical properties such as nonexpansivity or proximability~\cite{pesquet_learning_2021,hurault_proximal_2022}. However, these conditions are usually not satisfied by standard pretrained denoisers.

In summary, insisting on convergence guarantees restricts the range of usable denoisers. Therefore, we shift our focus from formal convergence to the more practical objective of maintaining stable PSNR over iterations. This was partly inspired by~\cite{terris2024equivariant}, which showed that making the denoiser invariant to transformations like rotation, reflection, and translation can help stabilize PnP and improve peak reconstruction quality. The resulting method, called \textbf{Equivariant-PnP}, can delay the PSNR drop often seen in standard PnP. However, in many cases, this only postpones the breakdown rather than preventing it altogether (see~\cref{fig:equi1}). While this delay can be helpful for early stopping strategies~\cite{wei2020tuning,effland2020variational}, a more reliable approach would be to remove the instability altogether.

\subsection{Contribution}

As seen in~\cref{fig:frontimage,fig:deep_den_instable}, even an unstable PnP system can give good results if the iterations are stopped when the reconstruction quality is at its best. We refer to this point as the \textbf{peak reconstruction}. However, figuring out exactly when to stop is not easy. A more practical approach would be to design a mechanism that automatically locks the output at this peak point (see~\cref{fig:frontimage}, second row). This would give us a stable system where the reconstruction quality does not degrade with iterations. Specifically, our goal is to develop a framework that achieves the following objectives:

\begin{enumerate}

   \item  The framework should stabilize any generic PnP system, regardless of the choice of proximal algorithm, pretrained denoiser, or the forward model. We use PSNR to measure reconstruction quality and consider a PnP system unstable if the PSNR noticeably drops over iterations (see~\cref{fig:deep_den_instable}). The system is considered stable if the PSNR remains steady or improves, without peaking and dropping sharply (see~\cref{fig:equi-diverge}).
   
    \item The framework should work without knowing the internal details of the black-box system. This is crucial for broad applicability, since the reasons for instability can vary widely and are often hard to identify.

\end{enumerate}

To achieve the above goals, we propose a practical, data-driven stabilization framework inspired by classical viscosity regularization~\cite{attouch_viscosity_1996}. In optimization problems with multiple solutions, viscosity regularization helps select a specific solution by adding a chosen regularizer. Likewise, in fixed-point methods, a contractive viscosity operator can guide the iterates toward a well-behaved fixed point~\cite{xu_viscosity_2004,sabach_first_2017}. However, PnP systems are not always tied to an objective function and may not have fixed points, so classical viscosity methods cannot be directly applied. Instead, we borrow the core idea of viscosity regularization to improve stability and prevent breakdowns. Specifically, we introduce a contractive stabilizing operator and design a mechanism that adaptively adjusts its influence to control instability and avoid PSNR drops.

\subsection{Organization}

In~\cref{sec:background}, we introduce the key concepts and background needed to understand our approach. In~\cref{sec:method}, we explain the motivation behind viscosity regularization and describe our stabilization algorithm. We present extensive experiments and comparisons in~\cref{sec:exps} to show that our method works well across different denoisers and proximal algorithms. Finally, we summarize our findings and discuss key insights and future directions in~\cref{sec:conclusion}.

\begin{figure}[t]
    \centering
    \includegraphics[width=0.75\columnwidth]{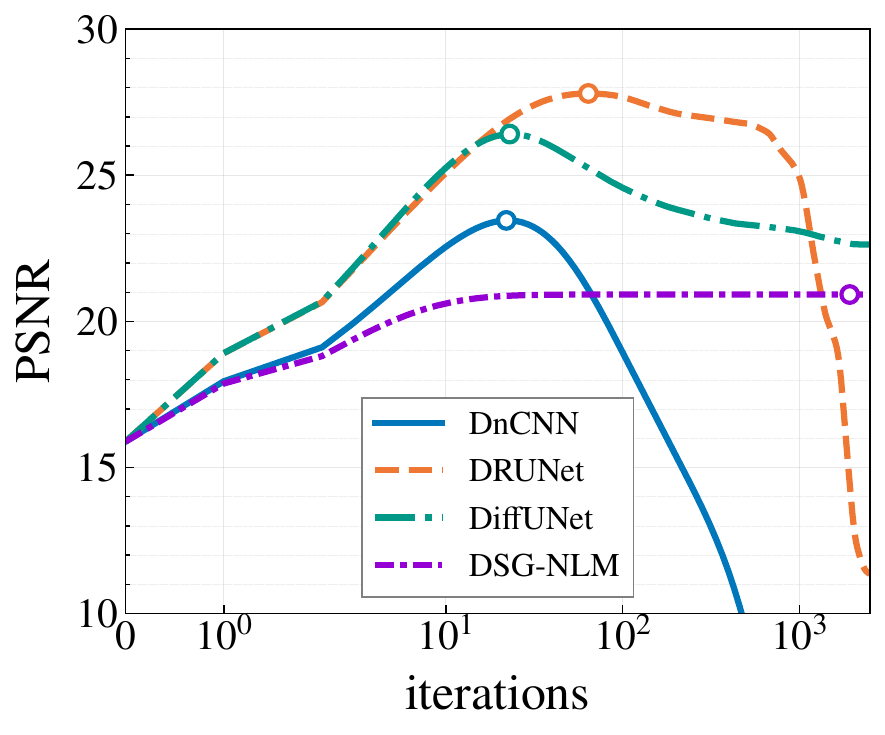}
    \caption{PSNR instability with various deep denoisers plugged into PnP-PGD~\eqref{eq:pnppgd}. Similar patterns of instability are observed with proximal algorithms such as PnP-HQS and PnP-ADMM~(\cref{sec:exps}). The deep denoisers shown here are trained for single-step denoising and therefore lack the structural properties to guarantee a stable PnP system. We also compare the PSNR trend obtained using a classical denoiser from~\cite{sreehari2016plug}, whose corresponding PnP system is known to be provably convergent~\cite{arghya-tsp}. This forms the foundation of our stabilization framework.}
    \label{fig:deep_den_instable}
\end{figure}

\section{Background}
\label{sec:background}

The key question is whether we can develop a convergence theory for the iterative process in~\eqref{eq:vanilla}, where \(\cT\) represents any of the operators in~\eqref{eq:Tpgd},~\eqref{eq:Thqs}, and~\eqref{eq:Tadmm}.  The main challenge is that~\eqref{eq:vanilla} may not correspond to minimizing any well-defined objective function when \(\cD\) is a pretrained denoiser. As a result, standard tools from optimization theory cannot be directly applied to analyze convergence~\cite{bauschke2019convex, beck2017book}. An alternative approach is to use ideas from operator theory~\cite{bauschke_fixed-point_2011}. First, we provide an operator-based description of PnP methods.

\subsection{PnP Operators}

As mentioned earlier, PnP models are built by replacing the proximal operator with a denoiser.  For example, if we make this substitution in the PGD algorithm, we obtain~\eqref{eq:pnppgd}, which we refer to as \textbf{PnP-PGD}. The same idea can be applied to algorithms such as HQS~\cite{geman1995nonlinear} and ADMM~\cite{bauschke2019convex}. We will refer to the corresponding PnP algorithms as \textbf{PnP-HQS} and \textbf{PnP-ADMM}~\cite{zhang2021plug,yuan2020plug,ryu2019plug,chan2016plug}. These algorithms have more detailed formulations~\cite{zhang2021plug,ryu2019plug}, but we only need to understand them as fixed-point operations for our purposes. Specifically, we can write \eqref{eq:pnppgd} as
\begin{equation}
\label{eq:Tpgd}
\x_{k+1}=\Tpgd(\x_k), \quad \Tpgd = \cD \circ \big( \cI-\gamma \nabla \!f  \big),
\end{equation}
where $\cI$ is the identity operator on $\Re^n$ and $\circ$ denotes composition of two operators. The corresponding operators in PnP-HQS and PnP-ADMM are
\begin{equation}
\label{eq:Thqs}
\Thqs = \cD \circ \mathrm{prox}_{\mu  f},
\end{equation}
and
\begin{equation}
\label{eq:Tadmm}
\Tadmm = \frac{1}{2} \left( \cI +  (2\cD - \cI ) \circ (2  \mathrm{prox}_{\alpha f} - \cI ) \right),
\end{equation}
where $f$ is the loss function in~\eqref{eq:fg} and $\mu, \alpha >0$ are tunable parameters. 

We generally refer to~\eqref{eq:Tpgd}, \eqref{eq:Thqs}, and~\eqref{eq:Tadmm}  as \textbf{Vanilla-PnP} operator, denoted by $\cT$. The corresponding iterations 
\begin{equation}
\label{eq:vanilla}
\begin{cases}
\x_0 \in \Re^n, \\
\x_{k+1} = \cT(\x_k), \qquad k \geqslant 0,
\end{cases}
\end{equation}
are referred to as \textbf{Vanilla-PnP}. The above identifications allow us to draw connections with operator theory. 

\subsection{Fixed-Point Convergence}

An operator $\cT: \Re^n \to \Re^n$ is said to be Lipschitz continuous if there exists $\beta > 0$ such that for all $\x, \y \in \Re^n$:
\begin{equation}
\label{eq:Lipschitz}
\|\cT(\x) - \cT(\y)\| \leqslant \beta \|\x - \y\|.
\end{equation}
The operator is nonexpansive if $\beta \leqslant 1$, and contractive (or a contraction) if $\beta < 1$. 
If the smallest possible $\beta$ in~\eqref{eq:Lipschitz} is greater than $1$, then the operator is called expansive.
A nonexpansive operator $\cT$ is said to be averaged (or $\theta$-averaged) if there exists a nonexpansive operator $\cN: \Re^n \to \Re^n$ and $\theta \in (0,1)$ such that $\cT = (1 - \theta)\cN + \theta \cI$.
In this case, the operators $\cT$ and $\cN$ have the same fixed points~\cite{bauschke2019convex}.

A key property of a contractive operator $\cT$ is that it has a unique fixed point $\p \in \Re^n$ such that $\cT(\p) = \p$. Moreover, for any $\x_0 \in \Re^n$, the iterations in~\eqref{eq:vanilla} are guaranteed to converge to $\p$. On the other hand, an averaged operator $\cT$ need not have any fixed points. However, if $\cT$ has a fixed point, the iterations are guaranteed to converge to a fixed point~\cite{bauschke_fixed-point_2011}.

The Vanilla-PnP operator has two components: one that comes from the denoiser, and another based on the loss function. The second part can be made nonexpansive by adjusting the parameters—e.g., the step size in PnP-PGD~\cite{ACK2023-contractivity}. Moreover, if the denoiser $\cD$ is averaged or contractive, fixed-point theory can guarantee convergence~\cite{bauschke2019convex}.  Some Vanilla-PnP operators are designed to be averaged~\cite{nair2024averaged, pesquet_learning_2021}, and can even be contractive~\cite{arghya-tsp}. In such cases, Vanilla-PnP is guaranteed to converge to a fixed point.

The challenge with pretrained denoisers such as DnCNN~\cite{zhang2017beyond}, DRUNet~\cite{zhang2021plug}, and DiffUNet~\cite{choi2021conditioning} is that, although they are Lipschitz continuous, they are usually expansive~\cite{ryu2019plug, hurault2022gradient, nair2024averaged}. In fact, testing whether a deep denoiser is nonexpansive is intractable, since computing the smallest $\beta$ in~\eqref{eq:Lipschitz} is known to be NP-hard for neural networks~\cite{virmaux2018lipschitz}. This makes it hard to train nonexpansive deep denoisers—it requires enforcing structural properties that are difficult even to test.

\begin{figure*}[thbp]
    \centering
    \subfloat{
    \includegraphics[width=0.49\linewidth]{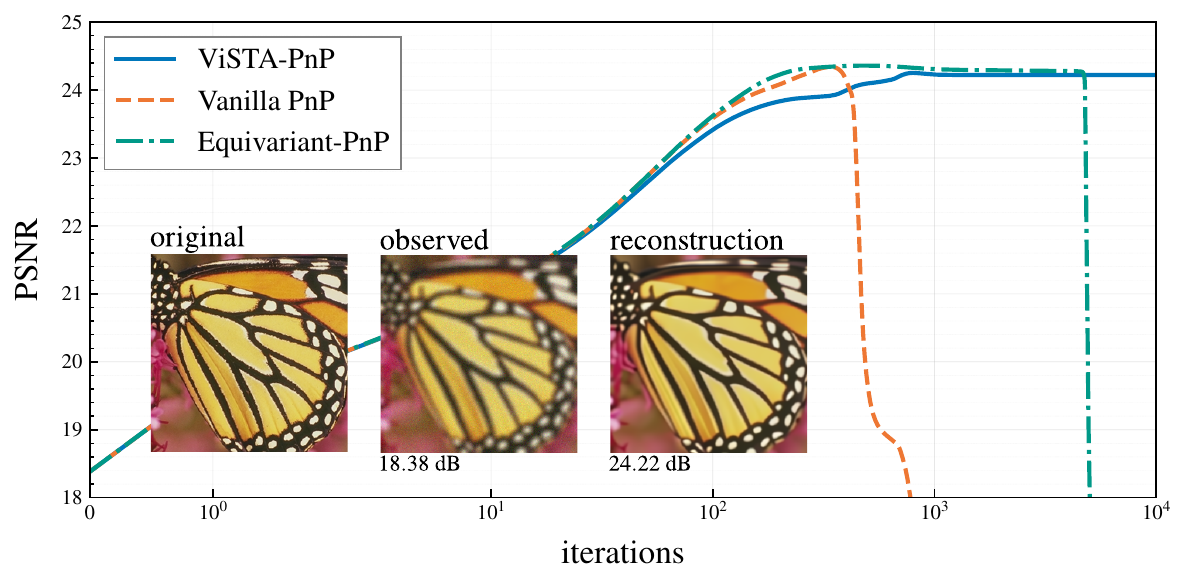}\label{fig:equi1}} \hfill
    \subfloat{
    \includegraphics[width=0.49\linewidth]{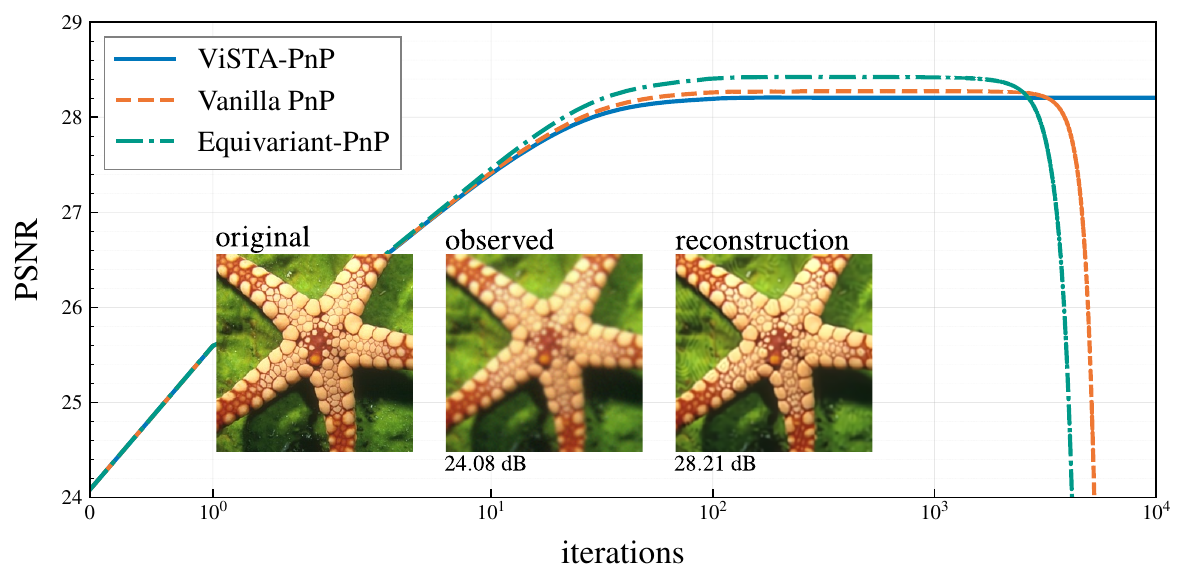}\label{fig:equi2}}
    \caption{Comparison of Vanilla-PnP, Equivariant-PnP, and ViSTA-PnP for Gaussian deblurring ($9\times 9$ Gaussian blur, standard deviation $4$, additive noise $0.03$) (\textbf{left}) and $2\times$ superresolution (\textbf{right}). We use PnP-PGD+DRUNet for both tasks, with test images drawn from the set3c dataset. Vanilla-PnP shows early divergence in the deblurring task, while Equivariant-PnP~\cite{terris2024equivariant} delays the breakdown but ultimately fails. In the superresolution case, both Vanilla-PnP and Equivariant-PnP exhibit early divergence. In contrast, ViSTA-PnP stabilizes the iterations effectively in both settings, without significantly compromising reconstruction quality. The peak PSNR is within $0.1$ dB (resp.~$0.2$ dB) of Vanilla-PnP (resp.~Equivariant-PnP).}
    \label{fig:equi-diverge}
\end{figure*}

\section{Method}
\label{sec:method}

\subsection{Viscosity Regularization}

The idea behind our approach comes from operator averaging. As noted earlier, if $\cT$ is nonexpansive, then it has the same fixed points as the averaged operator $(1 - \theta)\cT + \theta \cI$, where $\theta \in (0,1)$. Moreover, this averaged operator can be used to stably compute a fixed point of $\cT$. In fact, the fixed-point iterations~\eqref{eq:vanilla} may not converge when applied directly to $\cT$, but they can become stable when $\cT$ is replaced with its averaged version.

The challenge is that the Vanilla-PnP operator is expansive in most PnP systems that use deep denoisers~\cite{cohen2021has,terris2024equivariant}. This leads to a natural question: instead of averaging $\cT$ with the identity operator, could we benefit by averaging it with a stronger operator?  A promising idea is to use a contraction operator $\cS$, derived from a classical reconstruction algorithm~\cite{arghya-tsp}.

The motivation for choosing a contraction stems from the denoiser $\cD$ being Lipschitz for most neural network architectures~\cite{gouk2021regularisation}. Since the Vanilla-PnP operator $\cT$  involves compositions of $\cD$ with linear operators, $\cT$ is also Lipschitz.  In this regard, we have the following simple observation.

\begin{proposition}
\label{prop:averagingTS}
Suppose $\cT$ is a Lipschitz operator and $\cS$ is a contraction. Then, there exists $\theta_0 > 0$ such that the operator $(1 - \theta)\cT + \theta \cS$ is contractive for all $\theta \in (\theta_0, 1]$. In particular, the iterations
\begin{equation}
\label{eq:contra-avg}
\begin{cases}
\x_0 \in \Re^n, \\
\x_{k+1} =  (1 - \theta)\cT(\x_k) + \theta \cS(\x_k)  \qquad (k \geqslant 0),
\end{cases}
\end{equation}
are convergent for all $x_0 \in \Re^n$.
\end{proposition}

It is important to note that \cref{prop:averagingTS} does not necessarily hold when $\cS$ is merely nonexpansive rather than contractive. The challenges in directly applying~\eqref{eq:contra-avg} are:
\begin{enumerate}
\item[(i)] For PnP systems where $\cD$ is a neural network, computing the smallest possible $\beta$ in~\eqref{eq:Lipschitz} is generally intractable. Moreover, the estimated values of $\beta$ for deep networks are too loose to be useful in practice~\cite{virmaux2018lipschitz}. In short, we cannot reliably set $\theta_0$ in~\cref{prop:averagingTS}.

\item[(ii)] To ensure that $(1 - \theta)\cT + \theta \cS$ is contractive, $\theta$ needs to be set close to $1$. However, this causes the classical operator $\cS$ to dominate the reconstruction process, which reduces the impact of the more powerful operator $\cT$.
\end{enumerate}

Despite the above concerns, \cref{prop:averagingTS} serves as a helpful starting point, and we base our stabilization mechanism on it.
While it is inherently difficult to address point~(i), point~(ii) can be tackled by gradually reducing the influence of the contraction operator $\cS$ over the iterations. This can be done using the update rule
\begin{equation}
\label{eq:viscous_update}
\x_{k+1} = (1 - \theta_k) \cT(\x_k) + \theta_k \cS(\x_k),
\end{equation}
where the sequence $\{\theta_k\} \to 0$. The scheme~\eqref{eq:viscous_update} is known as the viscosity regularization of $\cT$, and the operator $\cS$ is referred to as the viscosity operator~\cite{xu_viscosity_2004,attouch_viscosity_1996}. We will refer to $\theta_k$ as the viscosity index.

Since $\theta_k$ changes with each iteration, the convergence of~\eqref{eq:viscous_update} is not easy to guarantee because standard results from the theory of averaged operators~\cite {bauschke2019convex} do not apply. Convergence in this setting was established in~\cite[Theorem 3.2]{xu_viscosity_2004}, and we present a simplified version of that result below.

\begin{theorem}
    \label{thm:viscosity}
    Suppose $\cT$ is nonexpansive and has fixed points, $\cS$ is contractive, and $\theta_k = 1/k$ in~\eqref{eq:viscous_update}. Then, for any $\x_0 \in \Re^n$, the iterates $\{\x_k\}$ in~\eqref{eq:viscous_update} converge to a fixed point of $\cT$.
\end{theorem}

Importantly, the limit point is determined by the viscosity operator $\cS$. In our context, this has a natural interpretation: the viscosity operator controls the quality of the reconstruction.

Unfortunately, we cannot directly apply~\cref{thm:viscosity}, as $\cT$ is typically expansive in PnP systems. Therefore, we use~\cref{thm:viscosity} as a guiding principle to promote stability rather than seeking a formal convergence guarantee. In particular, rather than letting $\theta_k \to 0$, we ensure that $\theta_k$ does not become too small. Indeed, if $\cS$ is completely phased out, the PnP iterations may again become unstable. To address this, we propose a data-driven strategy that adaptively sets the viscosity index while keeping it bounded away from zero.

\begin{figure*}
\centering
    \subfloat[instability with $\cS=\cI$]{
    \label{fig:theta-div1}
    \includegraphics[width=0.38\linewidth]{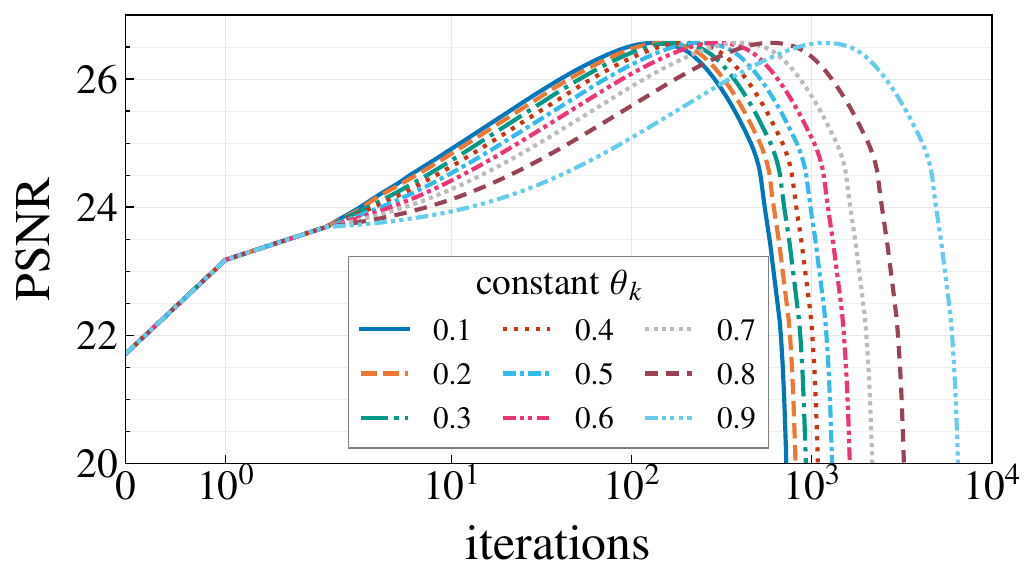}}
    \subfloat[$\cS=0.95\, \cI$ vs $\cS_{\mathrm{NLM}}$]{
    \label{fig:theta-div2}
    \includegraphics[width = 0.38\linewidth]{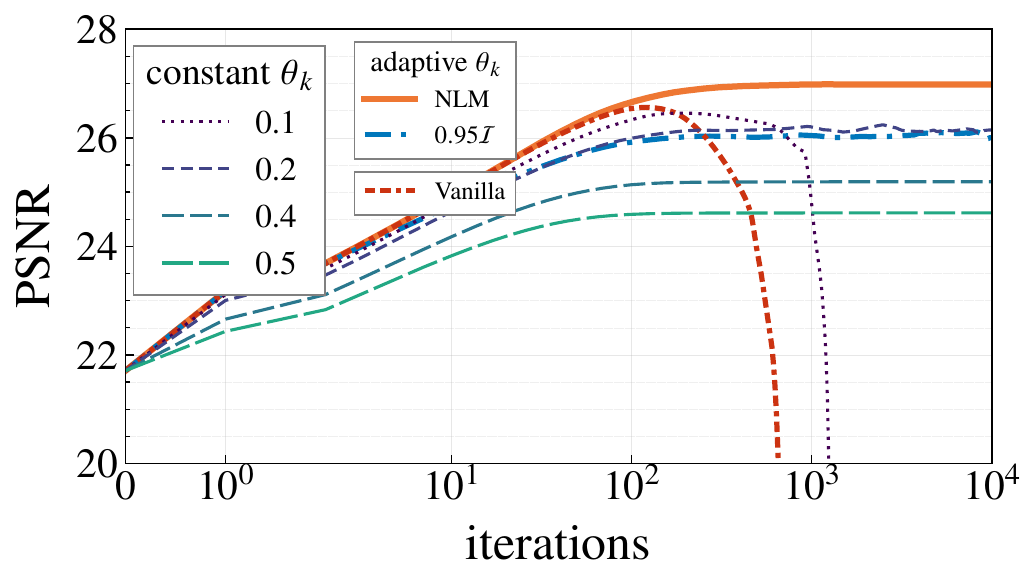}
    }
    \subfloat[evolution of adaptive $\theta_k$]{
    \label{fig:theta-div3}
    \includegraphics[width = 0.21\linewidth]{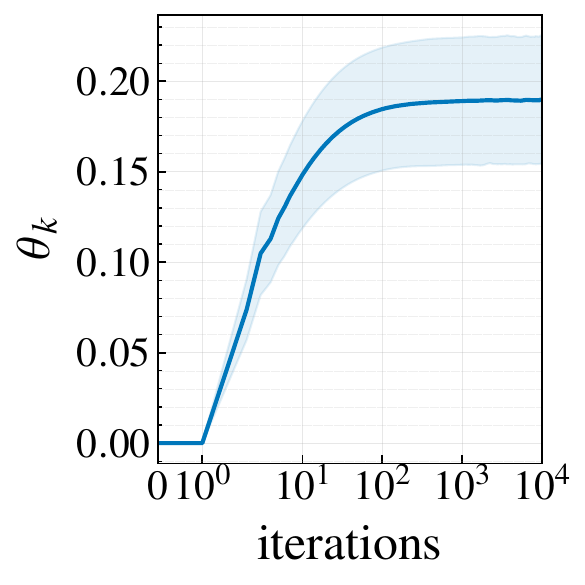}
    }
    \caption{Gaussian deblurring using the PnP-PGD+DRUNet combination. This experiment highlights the importance of choosing a contractive viscosity operator $\cS$. When we use $\cS = \cI$ (\textbf{left}), which is nonexpansive but not contractive, the iterations remain unstable, even with a fixed value of $\theta_k$. In contrast, using $\cS = 0.95, \cI$ makes the system stable for a fixed $\theta_k > 0.1$ (\textbf{center}), and also under the adaptive rule~\eqref{thetak} (\textbf{right}). As expected, increasing the viscosity index improves stability, though it may cause a slight drop in PSNR. Using the more effective contractive operator $\cS_{\mathrm{NLM}}$ in~\eqref{eq:NLM} gives better PSNR than $\cS = 0.95, \cI$. Under the adaptive rule~\eqref{thetak}, the viscosity index $\theta_k$ evolves and settles around $0.2$ on average across the set3c dataset. This strikes a good balance between the reconstruction operator $\cT$ and the stabilizing effect of $\cS$.}
    \label{fig:theta-div}
\end{figure*}

\subsection{Data-Driven Stabilization}
\label{sec:adaptive}

We describe a practical algorithm for selecting $\theta_k$ and discuss possible choices for the viscosity operator $\cS$.
Intuitively, the role of $\cS$ is to suppress small artifacts in the early stages of reconstruction, preventing them from growing and causing instability. However, if the viscosity index $\theta_k$ is set too high, it can overly stabilize the process too early, ultimately degrading the quality of the final reconstruction (see~\cref{fig:theta-div2}).

A practical strategy is to adaptively adjust $\theta_k$ based on the behavior of the past iterates, with the goal of keeping the sequence $\{\x_k\}$ bounded and preventing divergence. Since the viscosity operator $\cS$ is contractive, it admits a unique fixed point $\p$. Notably, $\cS$ can be derived from classical reconstruction methods, and its fixed point has a clear and interpretable meaning in practical applications. 

Since $\p$ is typically a good reconstruction, we can use this as a reference point to assess the stability of the iterates. To make this precise, we introduce the notion of an $\eta$-stable operator. Specifically, we say that $\cT$ is $\eta$-stable with respect to $\p \in \Re^n$ if there exists $\eta > 0$ such that
\begin{equation}
\label{eq:eta-stable}
\|\cT(\x) - \p\| \leqslant \eta \|\x - \p\| \qquad (\x \in \Re^n).
\end{equation}
The parameter $\eta$ quantifies how much the operator $\cT$ expands or contracts with respect to the fixed point $\p$. We note that this is distinct from quasi-nonexpansiveness~\cite{cohen_regularization_2021}.

The following observation and the subsequent discussion provide a guideline for adjusting the viscosity index.

\begin{proposition}
\label{prop:bound}
Let $\cS$ be a $\beta$-contraction with fixed point $\p$, and let $\cT$ be $\eta$-stable with respect to $\p$. If a constant viscosity $\theta_k = \theta$ is used, where $(\eta - 1)/(\eta - \beta) \leqslant \theta <1$ for $\eta > 1$, and $\theta = 0$ for $\eta \leqslant 1$, then the iterates $\{\x_k\}$ in~\eqref{eq:viscous_update} are bounded.
\end{proposition}

\begin{proof}
Note that the condition in~\cref{prop:bound} guarantees that $\theta \in [0,1]$. This is immediate when $\eta \leqslant 1$. For $\eta > 1$, we have $(\eta - 1)/(\eta - \beta) \in (0,1)$ since $\beta < 1 < \eta$. In particular, $(1-\theta)\eta + \theta \beta \leqslant 1$.  Since $\cS(\p)=\p$, we can write
    \begin{align*}
    \x_{k+1} - \p &=  (1-\theta) \left(\cT (\x_k) - \p \right) + \theta \left(\cS (\x_k) - \p\right) \\
    &= (1-\theta) \left(\cT (\x_k) - \p \right)  + \theta \big(\cS (\x_k) - \cS(\p)\big).
    \end{align*}
As $\cT$ is $\eta$-stable with respect to $\p$, $\cS$ is a $\beta$-contraction, and $\theta \in [0,1]$, it follows that
\begin{align*}
\|\x_{k+1} - \p\| & \leqslant  (1-\theta) \| \cT (\x_k) - \p\| +  \theta \| \cS(\x_k) -\cS(\p) \| \\
&\leqslant  \left( (1-\theta)\eta + \theta \beta \right) \|\x_k - \p\| \\
&\leqslant  \|\x_k - \p\|.
\end{align*}
By iterating the above bound, we obtain $\| \x_k - \p\| \leqslant \| \x_0 - \p\|$ for all $k \geqslant 1$, which establishes that  $ \{\x_k\} $ is bounded.
\end{proof}

As in~\cref{prop:averagingTS}, the contractivity of $\cS$  plays a crucial role in~\cref{prop:bound}. Of course, boundedness alone does not guarantee convergence or PSNR stability. Nonetheless, it is still a useful guarantee, especially considering the divergence of the iterates shown in~\cref{fig:equi-diverge}. In particular, ensuring boundedness helps prevent uncontrolled growth in the iterates, serving as a basic but important form of stabilization.

It is difficult to verify~\eqref{eq:eta-stable} in practice. Instead, motivated by~\cref{prop:bound}, we propose the following heuristic. At each iteration $k \geqslant 1$, we compute
\begin{equation}
    \label{etak}
    \eta_k = \frac{\|\cT(\x_k) - \p\|}{\|\x_k - \p\|} \quad \mbox{and} \quad \beta_k = \frac{\|\cS(\x_k) - \p \|}{ \|\x_k - \p\|},
\end{equation}
and set $\theta_k = (\eta_k - 1)/(\eta_k - \beta_k)$. Since $\beta_k \leqslant \beta < 1$, it follows that $\theta_k \in [0,1)$, as required by~\eqref{eq:viscous_update}. If $\eta_k \leqslant 1$, then by~\cref{prop:bound}, we set $\theta_k = 0$.

We also make the following adjustments. When $\theta_k$ approaches $1$, the influence of the Vanilla-PnP operator is reduced. To prevent this, we impose an upper bound $\Theta \in (0,1)$ and define
\begin{equation}
\label{thetak}
\theta_k = \min \left( \frac{\eta_k - 1}{\eta_k - \beta_k}, \Theta \right).
\end{equation}
Moreover, since~\eqref{etak} is undefined when $\x_k = \p$, we set $\theta_k = \Theta$ whenever $\x_k$ lies within a small neighborhood of $\p$.

A natural question is whether it is necessary for $\cS$ to be a contraction. In~\cref{fig:theta-div1}, we present an example showing that even when $\cS$ is nonexpansive, it may fail to prevent Vanilla-PnP from diverging if it is not strictly contractive.

The proposed stabilization framework, \textbf{Viscosity Stabilized PnP}~(ViSTA-PnP), is summarized in~\cref{algo:vista}.

\begin{algorithm}[H]
    \caption{ViSTA-PnP}
    \label{algo:vista}
    \begin{algorithmic}[1]
        \REQUIRE Vanilla-PnP $\cT$, contraction $\cS$, $\x_0$, and $\Theta$.
        \STATE compute the fixed point $\p$ of $\cS$.
        \FOR{$k = 0,1,\ldots$}
        \STATE set $\eta_k$ and $\beta_k$ using $\p, \cT\text{ and } \cS$ in~\eqref{etak}.
        \STATE set $\theta_k$ using $\Theta, \eta_k$ and $\beta_k$ in~\eqref{thetak}.
        \STATE update $\x_{k+1} = (1-\theta_k) \, \cT(\x_k) + \theta_k \, \cS (\x_k) $.
        \ENDFOR
    \end{algorithmic}
\end{algorithm}

\subsection{Viscosity Operator}
\label{sec:visc-op}

\begin{figure*}
    \centering
    \subfloat[original]{
    \includegraphics[width=0.31\columnwidth]{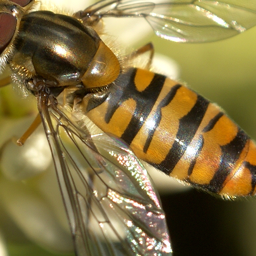}}\hfill
     \subfloat[observed]{
    \includegraphics[width=0.31\columnwidth]{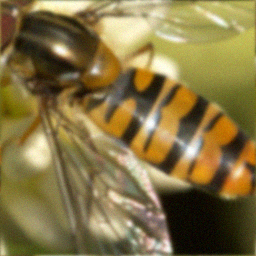}}\hfill
    \subfloat[Vanilla-PnP]{
    \includegraphics[width=0.31\columnwidth]{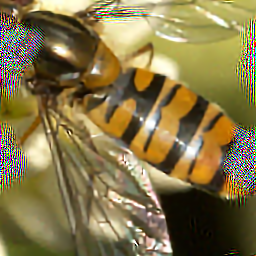}}\hfill
    \subfloat[ViSTA ($\,\cS = 0.95 \cI\,$)]{
    \includegraphics[width=0.31\columnwidth]{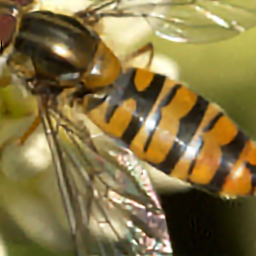}}\hfill
     \subfloat[ViSTA ($\cS_{\mathrm{NLM}}$)]{
    \includegraphics[width=0.31\columnwidth]{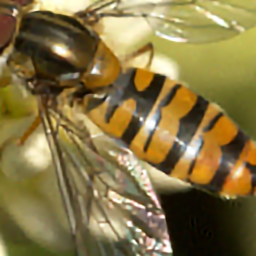}}\hfill
     \subfloat[MMO]{
    \includegraphics[width=0.31\columnwidth]{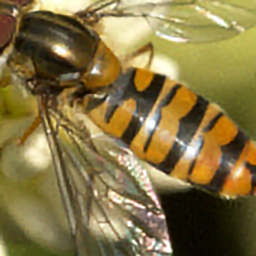}}
    \caption{Deblurring experiment on the \textit{wasp} image from the General100 dataset~\cite{general100}, using a $25 \times 25$ Gaussian blur (standard deviation $1.6$) and Gaussian noise of standard deviation $0.01$. We use the PnP-PGD+DnCNN  framework. We compare the performance of two viscosity operators, $\cS = 0.95 \cI$ and $\cS = \cS_{\mathrm{NLM}}$, within the ViSTA-PnP setup, against the baseline Vanilla-PnP. Additionally, we include results obtained using the Lip-DnCNN denoiser trained under the MMO framework~\cite{pesquet_learning_2021}. The PSNR values are: (b) $24.46$, (c) $-9.55$, (d) $28.81$, (e) $28.91$, and (f) $28.87$.}
    \label{fig:different_S}
\end{figure*}

As illustrated in~\cref{fig:different_S}, even the trivial contraction $\cS = \beta \cI$ with $\beta \in [0,1)$, can stabilize Vanilla-PnP. In this case, the fixed point $\p =0$. Although this may be effective in some situations, it often leads to poor reconstructions when used with powerful denoisers such as DRUNet~\cite{zhang2021plug}, DiffUNet~\cite{choi2021conditioning}, and Restormer~\cite{zamir2022restormer}.

A more promising strategy is to use a classical reconstruction operator $\cS$ that is naturally contractive and whose fixed point $\p$ is a high-quality image.
Surprisingly, such operators are rare in the literature. We propose using the PnP-PGD operator associated with the Non-Local Means (NLM) denoiser~\cite{buades2005non}, defined as
\begin{equation}
\label{eq:NLM}
\cS_{\mathrm{NLM}} = \cD_{\mathrm{NLM}} \circ (\cI - \rho \nabla \! f),
\end{equation}
where $\cD_{\mathrm{NLM}}$ is a proximable variant~\cite{sreehari2016plug} of the NLM denoiser. It was shown in~\cite[Theorem~1]{ACK2023-contractivity} that $\cS_{\mathrm{NLM}}$ is contractive for deblurring and superresolution tasks if $\rho \in (0,2)$.

The operator \(\cS_{\mathrm{NLM}}\) is known to give high-quality reconstruction~\cite{sreehari2016plug,ACK2023-contractivity}. To compute~\eqref{etak}, we need the fixed point of $\cS$, but a rough estimate is sufficient in practice. This can be done by applying $\cS$ a few times. The next sections show that $\cS_{\mathrm{NLM}}$ provides stability and strong reconstruction performance.

The naive NLM denoiser and its proximable variant $\cD_{\mathrm{NLM}}$ are computationally expensive, relying on nested loops over pixels, channels, and patches. To address this, we developed a PyTorch implementation that removes these loops by vectorizing the computation. Specifically, instead of serially processing the image pixels, we use GPU-friendly tensor operations (\texttt{unfold}, \texttt{roll}, and \texttt{batched shifts}) to process entire image windows at once on the GPU. This reduces the time complexity from $\cO(W^2 P^2 C N M)$ to $\cO(W^2)$, where $W$, $P$, $C$, $N$, and $M$ are respectively the search window size, patch size, number of channels, image height, and image width. The tradeoff, however, is higher memory usage: the VRAM requirement increases from $\cO(W^2)$ to $\cO(W^2 P^2 C N M)$ due to the need to store large intermediate tensors.

\section{Experiments}
\label{sec:exps}

We emphasize that ViSTA is not intended to improve reconstruction quality or outperform state-of-the-art methods. Instead, the goal is to mitigate the instability in PnP systems. We test ViSTA-PnP~(\cref{algo:vista}) on different reconstruction tasks, proximal algorithms, and denoisers to show that it works well in many situations and is a useful tool for improving stability. We use the notation A+D to represent the combination of a proximal algorithm A and a denoiser D (for example, PnP-PGD+DnCNN or PnP-HQS+DiffUNet). One of the main goals of the experiments is to check whether ViSTA-PnP can match or even improve upon the best performance of Vanilla-PnP.

\subsection{Experimental Setup}

We focus on two key reconstruction tasks: deblurring and superresolution~\cite{zhang2021plug,guo_mambair_2025,terris2024equivariant}. For motion deblurring, we use the blur kernels from~\cite{levin_kernel_2009}, while for Gaussian deblurring, we apply a $25 \times 25$ Gaussian kernel with standard deviation $1.6$. We use Gaussian blurring followed by either $2\times$ (or $4\times$) downsampling for superresolution. In all experiments, the noise term $\boldsymbol{\epsilon}$ in~\eqref{eq:fm} is modeled as Gaussian noise, with standard deviation in the $[0, 0.03]$ range.

The test images are from set3c, CBSD10~\cite{bsd500} urban100~\cite{urban100}, and General100~\cite{general100}. We use PnP-PGD, PnP-HQS, and PnP-ADMM as reconstruction algorithms.
As backbone denoisers, we use pretrained models of DnCNN~\cite{zhang2017beyond}, DRUNet~\cite{zhang2021plug}, DiffUNet~\cite{choi2021conditioning}, GSDRUNet~\cite{hurault2022gradient}, and MMO~\cite{pesquet_learning_2021} from the DeepInverse library~\cite{tachella2023deepinverse}. All experiments are performed on a single NVIDIA RTX A6000 GPU.

For the initialization $\x_0$ in~\cref{algo:vista}, we use the observed image for deblurring and its bicubic interpolation for superresolution. We primarily compare the stability and performance of ViSTA-PnP with Vanilla-PnP, Equivariant-PnP~\cite{terris2024equivariant}, and also with state-of-the-art methods: DPIR \cite{zhang2021plug}, DiffPIR~\cite{zhu_denoising_2023}, and GSPnP~\cite{hurault2022gradient}.

In~$\cS_{\mathrm{NLM}}$ we set the step size to $\rho = 1.9$ and the window size, patch size, and filtering parameter to $3, 3$, and $60/255$.  For Equivariant-PnP, we either average over the dihedral group $D_4$ (reflections and 90-degree rotations) or sample from it~\cite{terris2024equivariant}. For DPIR and GSPnP, we use the default parameters provided in~\cite{zhang2021plug} and ~\cite{hurault2022gradient}, respectively. 

We use \textbf{peak PSNR} for the highest PSNR over all iterations, and \textbf{asymptotic PSNR} for the PSNR after $10K$ iterations.

\begin{table*}[htpb]
\centering
\small
\begin{tabular}{lccccccc}
\toprule
\multirow{3}{*}{\textbf{Framwork}} & \multirow{3}{*}{\textbf{Method}} & \multicolumn{4}{c}{\textbf{Deblur}} & \multicolumn{2}{c}{\textbf{Superresolution}} \\
 &  & \multicolumn{2}{c}{\textbf{Gaussian}} & \multicolumn{2}{c}{\textbf{Motion} (Kernel 3~\cite{levin_kernel_2009})} & \multicolumn{2}{c}{$\mathbf{2\times}$} \\
 &  & \textbf{Peak} & \textbf{Asymtotic} & \textbf{Peak} & \textbf{Asymtotic} & \textbf{Peak} & \textbf{Asymtotic} \\
\midrule
\multicolumn{2}{c}{\textbf{Observed}} & \multicolumn{2}{c}{$24.15\pm3.19$} & \multicolumn{2}{c}{$21.86\pm3.37$} & \multicolumn{2}{c}{$23.76\pm3.39$} \\
\midrule
\multirow{3}{*}{\makecell[l]{\textit{PnP-PGD} + \\\textbf{DnCNN}~\cite{zhang2017beyond}}}
& \textbf{ViSTA}  & $27.84\pm4.22$ & $\mathbf{27.79\pm4.24}$ & $\mathbf{27.19\pm2.21}$ & $\mathbf{27.02\pm2.00}$ & $26.58\pm4.23$ & $\mathbf{26.55\pm4.18}$ \\
& \textbf{Equiv.} & $\mathbf{28.10\pm4.17}$ & \xmark & \uline{$26.80\pm2.32$} & \xmark & $\mathbf{27.08\pm4.19}$ & \uline{$25.39\pm5.97$} \\
& \textbf{Vanilla} & \uline{$27.99\pm4.20$} & \xmark & $26.76\pm2.33$ & \xmark & \uline{$26.97\pm4.24$} & \xmark \\
\midrule
\multirow{3}{*}{\makecell[l]{\textit{PnP-HQS} + \\\textbf{DRUNet}~\cite{zhang2021plug}}}
& \textbf{ViSTA}  & $27.69\pm4.42$ & $\mathbf{27.68\pm4.41}$ & $29.01\pm4.25$ & $\mathbf{29.01\pm4.25}$ & $26.72\pm4.27$ & \uline{$26.71\pm4.26$} \\
& \textbf{Equiv.} &
$\mathbf{27.88\pm4.40}$ & \uline{$27.58\pm4.16$} & $\mathbf{29.77\pm4.31}$ & \uline{$26.20\pm7.88$} & $\mathbf{26.93\pm4.29}$ & $\mathbf{26.82\pm4.29}$ \\
& \textbf{Vanilla} & \uline{$27.75\pm4.49$} & $26.76\pm5.17$ & \uline{$29.66\pm4.33$} & $27.46\pm6.54$ & \uline{$26.80\pm4.35$} & \xmark \\
\midrule
\multirow{3}{*}{\makecell[l]{\textit{PnP-ADMM} + \\ \textbf{DRUNet}~\cite{zhang2021plug}}} & \textbf{ViSTA} & \uline{$27.45\pm4.15$} & $\mathbf{27.45\pm4.15}$ & $28.57\pm3.63$ & $\mathbf{28.56\pm3.62}$ & \uline{$26.56\pm3.97$} & \uline{$26.55\pm3.97$} \\
 & \textbf{Equiv.} & $\mathbf{27.62\pm4.12}$ & \uline{$18.29\pm26.69$} & $\mathbf{29.01\pm3.57}$ & \uline{$20.52\pm8.51$} & $\mathbf{26.75\pm4.00}$ & $\mathbf{26.69\pm4.02}$ \\
 & \textbf{Vanilla} & $27.44\pm4.16$ & \xmark & \uline{$28.86\pm3.54$} & $16.04\pm14.04$ & $26.54\pm3.99$ & \xmark \\
\midrule
\multirow{3}{*}{\makecell[l]{\textit{PnP-HQS} + \\\textbf{DiffUNet}~\cite{choi2021conditioning}}}
& \textbf{ViSTA}   & $\mathbf{27.78\pm3.93}$ & $\mathbf{27.15\pm3.31}$ & $29.21\pm3.60$ & $\mathbf{29.00\pm3.56}$ & \uline{$26.59\pm3.74$} & $\mathbf{25.67\pm3.30}$ \\
& \textbf{Equiv.} & $27.77\pm3.81$ & \uline{$21.63\pm2.27$} & $\mathbf{29.54\pm3.46}$ & \uline{$25.87\pm3.47$} & $\mathbf{26.71\pm3.78}$ & \uline{$20.26\pm1.52$} \\
 & \textbf{Vanilla} & $27.55\pm3.68$ & $19.25\pm1.11$ & \uline{$29.38\pm3.39$} & $21.64\pm1.73$ & $26.51\pm3.64$ & $18.14\pm0.71$ \\
\midrule
\multirow{3}{*}{\makecell[l]{\textit{PnP-HQS} + \\\textbf{GSDRUNet}~\cite{hurault2022gradient}}} & \textbf{ViSTA} & $28.26\pm4.42$ & $\mathbf{28.22\pm4.38}$ & $30.47\pm4.38$ & $\mathbf{30.43\pm4.39}$ & $27.24\pm4.32$ & $\mathbf{27.23\pm4.31}$ \\
 & \textbf{Equiv.} & $\mathbf{28.38\pm4.50}$ & \uline{$26.59\pm4.64$} & $\mathbf{30.57\pm4.41}$ & \uline{$29.54\pm5.88$} & $\mathbf{27.34\pm4.35}$ & \uline{$25.78\pm4.20$} \\
 & \textbf{Vanilla} & \uline{$28.35\pm4.51$} & $24.89\pm6.55$ & \uline{$30.49\pm4.40$} & $23.98\pm9.15$ & \uline{$27.32\pm4.36$} & $24.75\pm7.41$ \\
\midrule
\textbf{DPIR}~\cite{zhang2021plug} & \textbf{Vanilla} & \multicolumn{2}{c}{$ 28.05\pm4.70 $} & \multicolumn{2}{c}{$ 29.55\pm4.52 $} & \multicolumn{2}{c}{$ 27.02\pm4.43 $} \\
\textbf{DiffPIR}~\cite{zhu_denoising_2023} & \textbf{Vanilla} & \multicolumn{2}{c}{ $ 27.84\pm3.67 $} & \multicolumn{2}{c}{$ 29.19\pm3.48 $} & \multicolumn{2}{c}{$ 27.01\pm3.88 $} \\
\textbf{GSPnP}~\cite{hurault2022gradient}  & \textbf{Vanilla} & \multicolumn{2}{c}{$ 28.92\pm4.61 $} & \multicolumn{2}{c}{$ 30.99\pm4.39 $} & \multicolumn{2}{c}{$ 27.38\pm4.48 $} \\
\bottomrule
\end{tabular}
\caption{PSNR results (mean $\pm$ std. dev.) for various applications using the PnP-PGD, PnP-HQS, and PnP-ADMM frameworks, averaged over the CBSD10 dataset. Gaussian noise with a standard deviation of $2\%$ was added. We compare against state-of-the-art methods: DPIR, DiffPIR, and GSPnP. \xmark~indicates divergence of the iterates.}
\label{tab:cnn_psnr_results}
\end{table*}

\begin{figure*}[ht]
    \centering
    \subfloat[original]{
    \includegraphics[width=0.32\columnwidth]{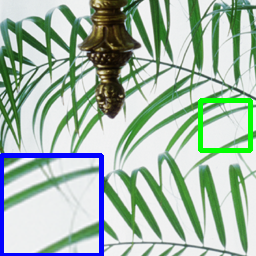}}\hfill
      \subfloat[observed ($12.51$)]{
    \includegraphics[width=0.32\columnwidth]{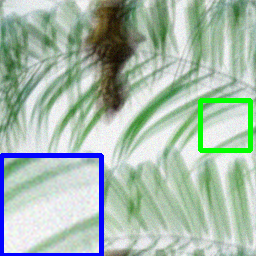}}\hfill
      \subfloat[Vanilla-PnP ($24.82$)]{
    \includegraphics[width=0.32\columnwidth]{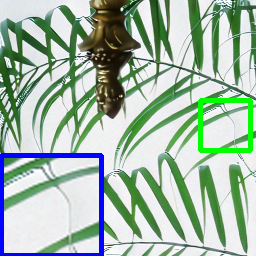}}\hfill
    \subfloat[DPIR ($28.87$)]{
    \includegraphics[width=0.32\columnwidth]{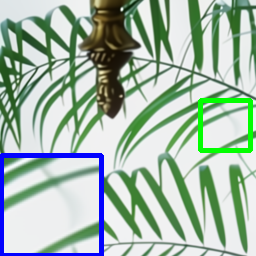}}\hfill
    \subfloat[GSPnP ($30.39$)]{
    \includegraphics[width=0.32\columnwidth]{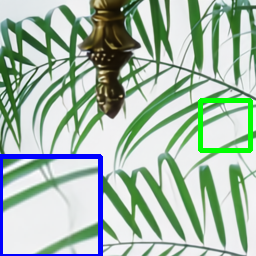}}\hfill
      \subfloat[ViSTA ($29.19$)]{
    \includegraphics[width=0.32\columnwidth]{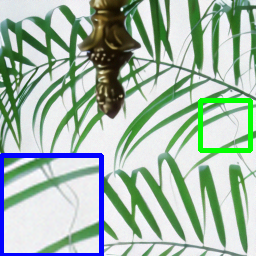}}\hfill

    \caption{Motion deblurring results on the \textit{leaves} image from the set3c dataset, using PnP-HQS+DiffUNet framework. We used kernel 8 from~\cite{levin_kernel_2009} and added Gaussian noise with standard deviation $0.03$. For ViSTA-PnP, we use a fixed $\theta_k = 0.05$ and the viscosity operator $\cS_{\mathrm{NLM}}$. The added viscosity improved the peak PSNR by $1.5$ dB over Vanilla-PnP and $1$ dB over Equivariant-PnP (see~\cref{fig:diff-psnr}). In particular, the visual quality  for ViSTA-PnP is better than DPIR~\cite{zhang2021plug} and GSPnP~\cite{hurault2022gradient}.}
    \label{fig:diffunet-vs-dpir}
\end{figure*}

\begin{figure}[ht]
    \centering
    \includegraphics[width=0.8\columnwidth]{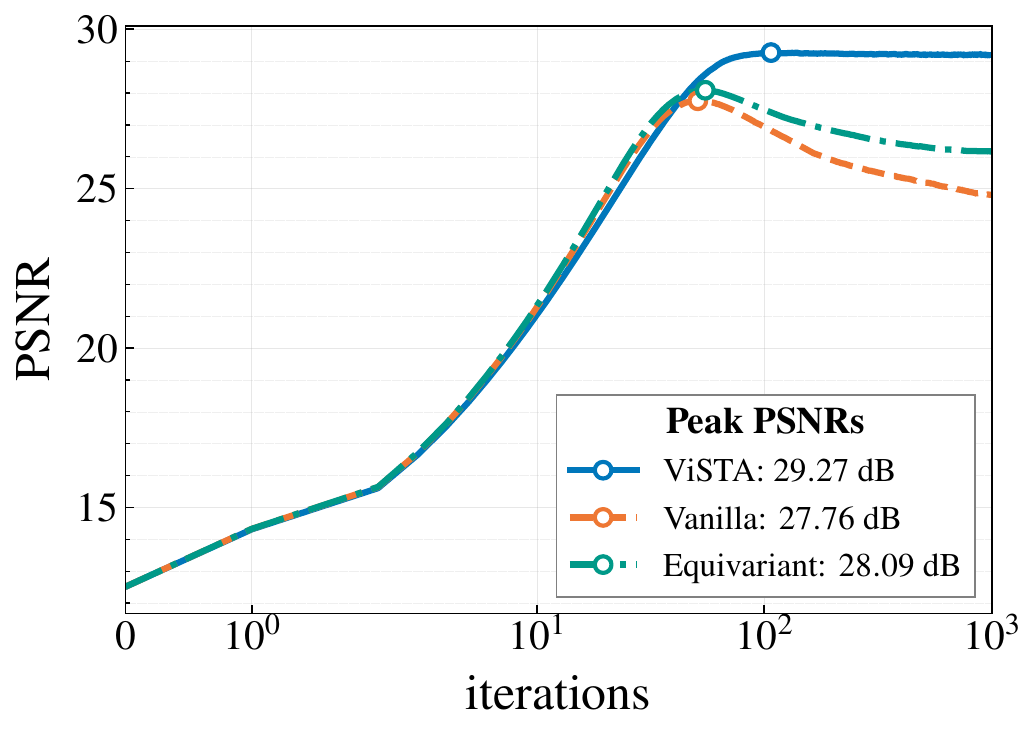}
    \caption{PSNR plot for the experiment in~\cref{fig:diffunet-vs-dpir}.}
    \label{fig:diff-psnr}
\end{figure}

\begin{figure*}[ht]
    \centering
    \subfloat[original]{
    \includegraphics[width=0.19\linewidth]{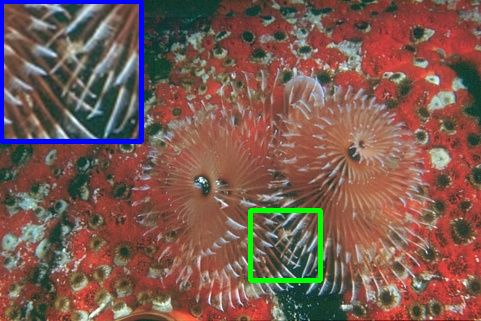}}\hfill
      \subfloat[observed ($22.69$)]{
    \includegraphics[width=0.19\linewidth]{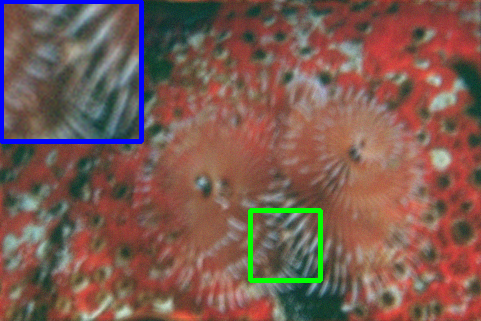}}\hfill
      \subfloat[Vanilla-PnP ($21.80$)]{
    \includegraphics[width=0.19\linewidth]{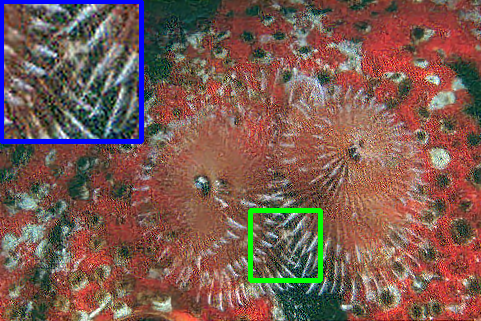}}\hfill
    \subfloat[Equivariant-PnP ($22.00$)]{
    \includegraphics[width=0.19\linewidth]{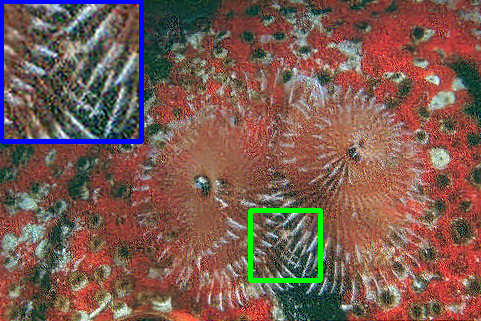}}\hfill
    \subfloat[ViSTA ($28.50$)]{
    \includegraphics[width=0.19\linewidth]{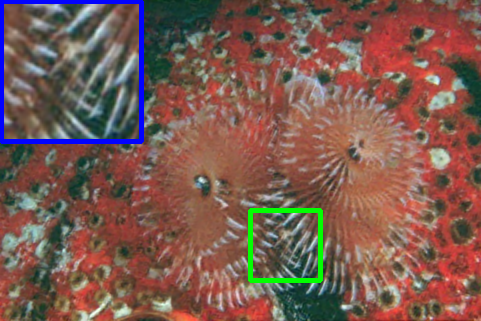}}\hfill

    \caption{Motion deblurring results on the \textit{coral} image from the CBSD68 dataset, using blur kernel 3 from~\cite{levin_kernel_2009} and additive Gaussian noise with standard deviation $0.02$. Reconstruction was performed using the PnP-PGD+MMO framework. For ViSTA-PnP, we set $\Theta = 0.02$ and used the viscosity operator $\cS_{\mathrm{NLM}}$. ViSTA-PnP successfully stabilizes the reconstruction around the peak PSNR, whereas the PSNR for Vanilla-PnP and Equivariant-PnP drops sharply by $6.7$ dB and $6.5$ dB.}
    \label{fig:mmo_vista_vs_equiv}
\end{figure*}

\begin{figure}[ht]
    \centering
    \includegraphics[width=0.8\columnwidth]{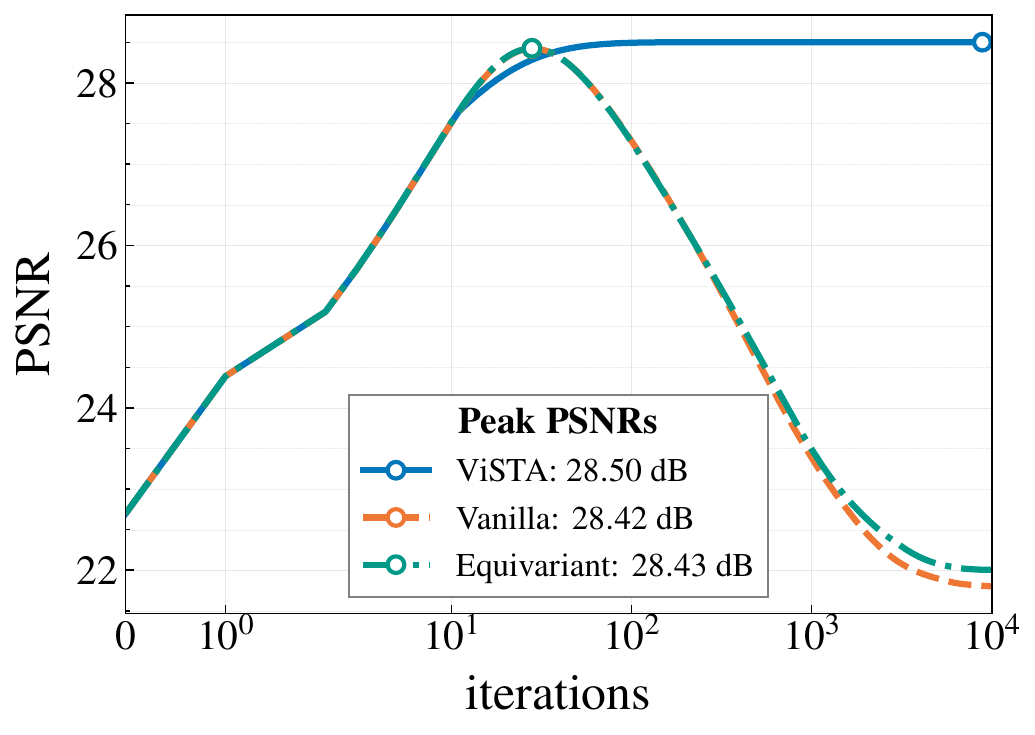}
    \caption{PSNR plot for the experiment in~\cref{fig:mmo_vista_vs_equiv}.}
    \label{fig:mmo-psnr}
\end{figure}

\begin{figure*}[ht]
    \centering
    \subfloat[original]{
    \includegraphics[width=0.32\columnwidth]{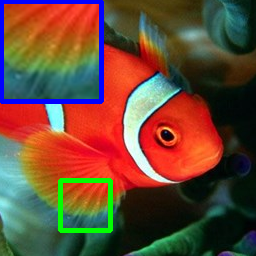}}\hfill
    \subfloat[bicubic]{
    \includegraphics[width=0.32\columnwidth]{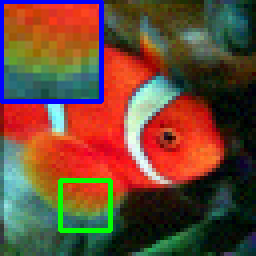}}\hfill
    \subfloat[fixed point $\p$]{
    \includegraphics[width=0.32\columnwidth]{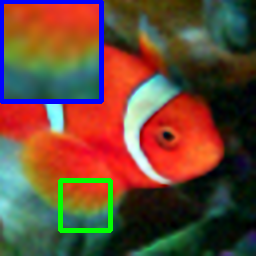}}\hfill
    \subfloat[DPIR]{
    \includegraphics[width=0.32\columnwidth]{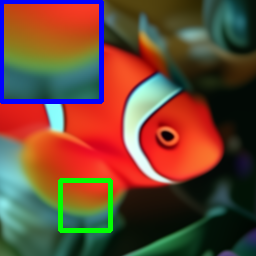}}\hfill
    \subfloat[GSPnP]{
    \includegraphics[width=0.32\columnwidth]{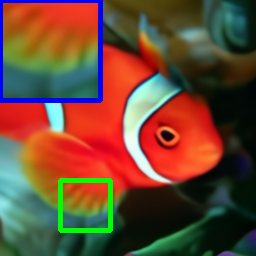}}\hfill
    \subfloat[ViSTA]{
    \includegraphics[width=0.32\columnwidth]{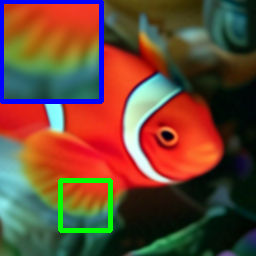}}\hfill
    \caption{Reconstruction results for $4\times$ superresolution (additive noise $0.03$) on the \textit{fish} image from the General100 dataset~\cite{general100}, using PnP-ADMM+DRUNet framework. The viscosity operator $\cS_{\mathrm{NLM}}$ was used in ViSTA-PnP. Also shown is the unique fixed point of the operator, which is used in~\eqref{etak} to compute the viscosity index. For comparison, results from DPIR~\cite{zhang2021plug} and GSPnP~\cite{hurault2022gradient} are included. The corresponding PSNR values are: (b) $22.71$, (c) $26.69$, (d) $29.28$, (e) $30.06$  and (f) $29.07$.}
    \label{fig:drunet_admm_4xSR}
\end{figure*}

\begin{figure}[ht]
    \centering
    \includegraphics[width=0.8\columnwidth]{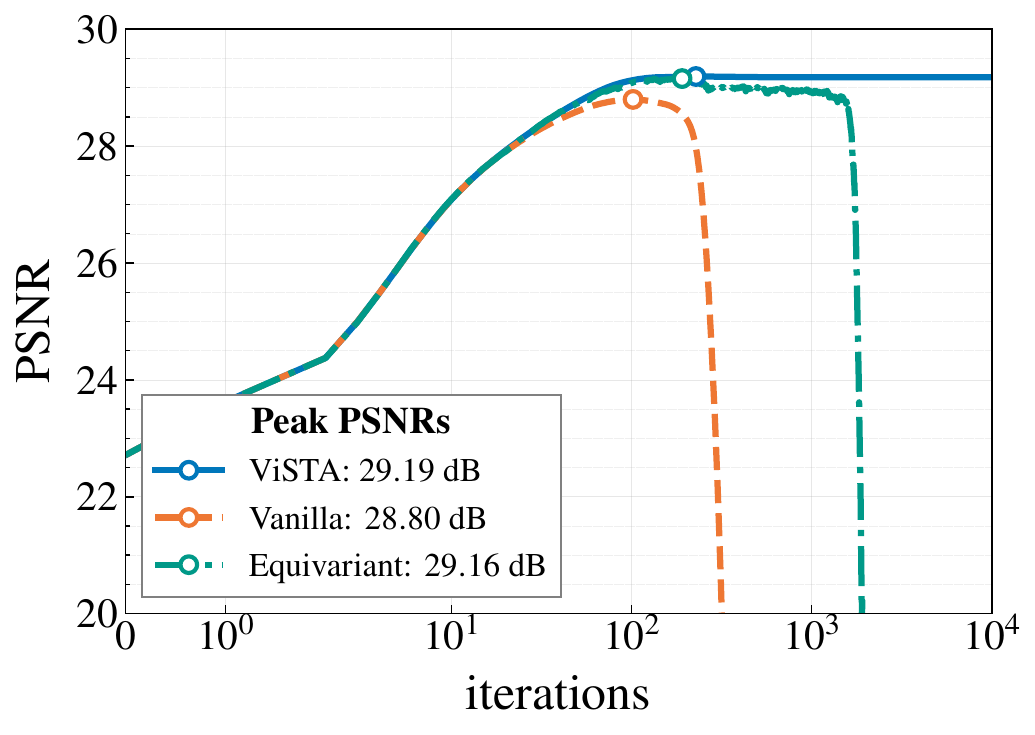}
    \caption{PSNR plot for the experiment in~\cref{fig:drunet_admm_4xSR}.}
    \label{fig:dru-psnr}
\end{figure}

\subsection{Stability Analysis}
We analyze how stability influences PSNR across iterations. Ideally, we aim to retain the peak PSNR while eliminating the sensitivity to stopping time. In~\cref{tab:cnn_psnr_results,tab:vit_psnr_results}, we compare the peak and asymptotic PSNRs in three applications, averaged on the CBSD10 data set. For ViSTA, we use the viscosity function defined in~$\cS_{\mathrm{NLM}}$, with the following values of $\Theta$ in~\eqref{thetak}: DnCNN (0.01), MMO (0.02), DRUNet (0.1), GSDRUNet (0.02), DiffUNet (0.2), Restormer (0.5), and SCUNet (0.5). A fixed value of $\Theta$ works well for a given denoiser in different applications and base algorithms, although fine-tuning can improve performance.

We use DnCNN and MMO as blind denoisers trained on noise levels in $[0, 2]/255$. DRUNet, GSDRUNet, and DiffUNet are used as nonblind denoisers, trained over the broader range $[0, 50]/255$ and evaluated at a fixed noise level of $\sigma = 5/255$. Restormer and SCUNet are blind denoisers trained on the $[0, 50]/255$ range. For Equivariant-PnP, we adopt the random sampling variant.

We find that ViSTA either improves the final PSNR or shows only a slight drop compared to Vanilla-PnP and Equivariant-PnP.  In contrast, Vanilla-PnP and Equivariant-PnP often fail to maintain their peak performance. As shown in~\cref{tab:cnn_psnr_results}, several Equivariant-PnP and Vanilla-PnP configurations diverge completely (marked with \xmark), highlighting their sensitivity to denoiser and optimization dynamics.  In contrast, ViSTA never diverges in any setting and delivers competitive or superior PSNR values.

\subsection{Robustness across Denoisers}

For the diffusion-based DiffUNet denoiser and the potential-based GSDRUNet denoiser, ViSTA achieves the best or near-best asymptotic scores while maintaining close alignment with peak values. This aligns well with the goal of reducing sensitivity to early stopping (see \cref{fig:diff-psnr}). Importantly, these results confirm that the viscosity-based regularization in ViSTA can enhance robustness not only under ideal conditions but also in challenging regimes where other methods fail.

\begin{figure}[t]
    \centering
    \subfloat[original]{
        \includegraphics[width=.31\columnwidth]{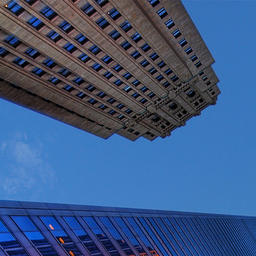}}\hfill
    \subfloat[Observed]{
        \includegraphics[width=.31\columnwidth]{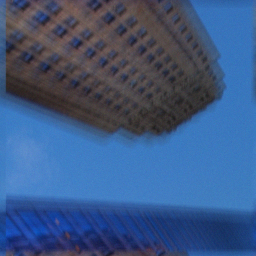}}\hfill
    \subfloat[DRUNet (ViSTA)]{
        \includegraphics[width=.31\columnwidth]{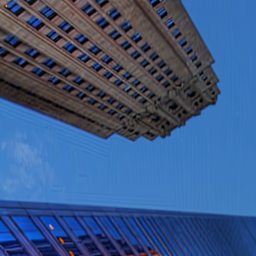}}\hfill
    \subfloat[Restormer (ViSTA)]{
        \includegraphics[width=.31\columnwidth]{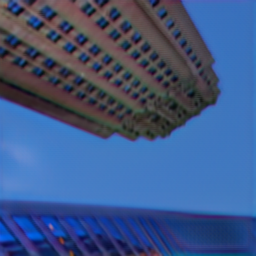}}\hfill
    \subfloat[SCUNet (ViSTA)]{
        \includegraphics[width=.31\columnwidth]{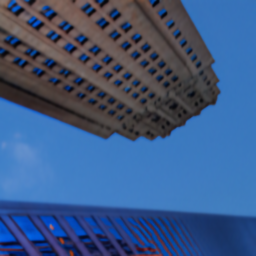}}\hfill
    \subfloat[Restormer (Equiv.)]{
        \includegraphics[width=.31\columnwidth]{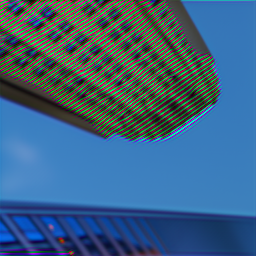}}\hfill
    \subfloat[SCUNet (Equiv.)]{
        \includegraphics[width=.31\columnwidth]{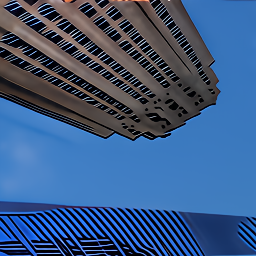}}\hfill
    \subfloat[Restormer (Vanilla)]{
        \includegraphics[width=.31\columnwidth]{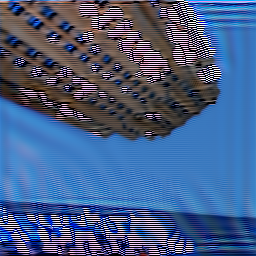}}\hfill
    \subfloat[SCUNet (Vanilla)]{
        \includegraphics[width=.31\columnwidth]{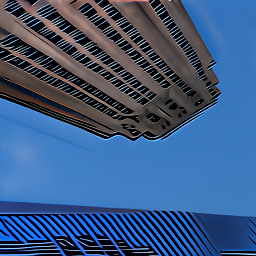}}\hfill
    \caption{Motion deblurring results on the \textit{skyscraper} image from the Urban100 dataset~\cite{urban100}, using blur kernel 7~\cite{levin_kernel_2009} and additive noise level $0.01$. We compare Vanilla-PnP, Equivariant-PnP, and ViSTA-PnP, using PnP-HQS+Restormer~\cite{zamir2022restormer} and PnP-HQS+SCUNet~\cite{zhang2023practical} frameworks. PSNR values: (b) $20.91$, (c) $28.31$, (d) $24.92$, (e) $25.17$, (f) $17.46$, (g) $16.86$, (h) $-8.44$, (i) $16.29$.}
    \label{fig:pnp_trans_den_motion_deblur}
\end{figure}

\begin{table*}[htpb]
\centering
\small
\begin{tabular}{lccccccc}
\toprule
\multirow{3}{*}{\textbf{Framwork}} & \multirow{3}{*}{\textbf{Method}} & \multicolumn{4}{c}{\textbf{Deblur}} & \multicolumn{2}{c}{\textbf{Superresolution}} \\
 &  & \multicolumn{2}{c}{\textbf{Gaussian}} & \multicolumn{2}{c}{\textbf{Motion} (Kernel 3~\cite{levin_kernel_2009})} & \multicolumn{2}{c}{$\mathbf{2\times}$} \\
 &  & \textbf{Peak} & \textbf{Asymtotic} & \textbf{Peak} & \textbf{Asymtotic} & \textbf{Peak} & \textbf{Asymtotic} \\
\midrule
\multicolumn{2}{c}{\textbf{Observed}} & \multicolumn{2}{c}{$24.15\pm3.19$} & \multicolumn{2}{c}{$21.86\pm3.37$} & \multicolumn{2}{c}{$23.76\pm3.39$} \\
\midrule
\multirow{3}{*}{\makecell[l]{\textit{PnP-PGD} + \\\textbf{Restormer}~\cite{zamir2022restormer}}} & \textbf{ViSTA} & $\mathbf{27.29\pm3.73}$ & $\mathbf{26.87\pm3.13}$ & $\mathbf{27.06\pm2.95}$ & $ \mathbf{26.74\pm2.40} $ & $ \mathbf{25.73\pm3.51} $ & $ \mathbf{25.49\pm3.64} $ \\
 & \textbf{Equiv.} & \uline{$ 27.21\pm3.76 $} & \uline{$ 9.14\pm5.49 $} & \uline{$ 26.63\pm3.05 $} & \uline{$ 12.56\pm6.35 $} & \uline{$ 25.16\pm3.06 $} & \uline{$ 3.91\pm4.60 $}  \\
 & \textbf{Vanilla} & $ 27.21\pm3.76 $ & $ 8.08\pm5.36 $ & $ 26.62\pm3.05 $ & $ 9.94\pm6.10 $ & $ 25.13\pm3.07 $ & $ 7.91\pm6.72 $  \\
\midrule
\multirow{3}{*}{\makecell[l]{\textit{PnP-PGD} + \\\textbf{SCUNet}~\cite{zhang2023practical}}} & \textbf{ViSTA} & $ \mathbf{27.81\pm4.87} $ & $ \mathbf{27.44\pm4.50} $ & $ \mathbf{27.33\pm4.68} $ & $ \mathbf{26.72\pm4.23} $ & $ \mathbf{25.82\pm4.58} $ & $ \mathbf{25.35\pm4.02} $ \\
 & \textbf{Equiv.} & \uline{$ 27.76\pm4.90 $} & \uline{$ 10.75\pm1.46 $} & \uline{$ 27.36\pm4.43 $} & \uline{$ 10.99\pm1.61 $} & \uline{$ 25.53\pm4.64 $} & \uline{$ 10.14\pm0.97 $}  \\
 & \textbf{Vanilla} & $ 27.75\pm4.86 $ & $ 10.73\pm1.43 $ & $ 27.28\pm4.40 $ & $ 10.52\pm1.25 $ & $ 25.55\pm4.66 $ & $ 10.28\pm1.20 $  \\
\bottomrule
\end{tabular}
\caption{PSNR results (mean $\pm$ std. dev.) for various applications using the PnP-PGD framework with Transformer-based denoisers, averaged over the CBSD10 dataset. Gaussian noise with a standard deviation of $2\%$ was added.}
\label{tab:vit_psnr_results}
\end{table*}

\begin{table}[ht]
\centering
\setlength{\tabcolsep}{2pt}
\scriptsize
\begin{tabular}{@{}cccccc@{}}
\toprule
& \textbf{Start} & & \textbf{Vanilla} & \textbf{Equivariant} & \textbf{ViSTA}  \\
\midrule
\multirow{2}{*}{\textbf{Deblur}} & \multirow{2}{*}{$ 22.13\pm3.31 $} & \textbf{Peak} & $ 28.59\pm3.16 $ & \uline{$28.60\pm3.16$} & $ \mathbf{28.74\pm3.30}$ \\
& & \textbf{Asymp.} & $ 24.22\pm4.08 $ & \uline{$26.16\pm3.48$} & $\mathbf{28.74\pm3.30}$\\
\midrule
\multirow{2}{*}{$\mathbf{2\times}$~\textbf{SR}} & \multirow{2}{*}{$ 23.76\pm3.39 $} & \textbf{Peak} & $ \mathbf{26.98\pm3.88} $ & $\mathbf{26.98\pm3.88}$ & $ 26.88\pm3.87 $\\
& & \textbf{Asymp.} & \uline{$ 26.93\pm3.87 $} &  $\mathbf{26.98\pm3.88}$ & $ 26.88\pm3.87 $\\
\bottomrule
\end{tabular}
\caption{PSNR results (mean $\pm$ standard deviation) on the CBSD10 dataset for $2\times$ superresolution and motion deblurring using kernel 3 from~\cite{levin_kernel_2009}, both with the PnP-PGD+MMO framework.}
\label{tab:mmo_psnr_results}
\end{table}

We next test the performance of ViSTA with the Lipschitz-constrained deep denoiser MMO~\cite{pesquet_learning_2021}. The results are shown in~\cref{tab:mmo_psnr_results}. We see that even when the underlying Vanilla-PnP method is inherently stable, introducing viscosity does not compromise its stability. However, in cases where the PSNR tends to drop over iterations (see~\cref{fig:mmo-psnr}), viscosity helps preserve the peak PSNR, leading to more stable performance.

In~\cref{fig:diffunet-vs-dpir}, we examine the effect of using a fixed viscosity parameter $\theta_k$ across all iterations in a motion deblurring experiment on the set3c images with PnP-HQS+DiffUNet. Applying a fixed viscosity $\theta_k = 0.05$ stabilizes the iterates under these conditions, yielding reconstructions that significantly outperform Vanilla-PnP and DPIR. In~\cref{fig:drunet_admm_4xSR}, we perform $4\times$ superresolution using PnP-ADMM+DRUNet. As seen in~\cref{fig:dru-psnr}, Equivariant-PnP again only delays the breakdown without fully preventing it. Although DPIR achieves a higher PSNR value than ViSTA-PnP in this case, the highlighted region reveals a lack of detail in its reconstruction, highlighting the qualitative advantage of ViSTA.

ViSTA can also stabilize PnP systems that use transformer-based denoisers such as Restormer~\cite{zamir2022restormer} and SCUNet~\cite{zhang2023practical}. ViSTA significantly reduces hallucination and checkerboard artifacts visible in Vanilla-PnP and Equivariant-PnP.  In fact, Equivariant-PnP fails entirely with Restormer and SCUNet in this setting, producing heavily distorted outputs and even negative PSNR values (e.g., $-8.44$ dB for Equivariant-PnP with SCUNet).

In contrast, ViSTA yields consistent and high-quality reconstructions across denoisers, achieving PSNR values of $24.92$~dB with Restormer and $25.17$~dB with SCUNet. These results are also reflected in the quantitative scores in~\cref{tab:vit_psnr_results}, where Equivariant-PnP and Vanilla-PnP either degrade significantly or fail entirely, whereas ViSTA achieves stability with high asymptotic PSNRs in different types of degradation. This demonstrates that ViSTA improves stability in diffusion-based models and generalizes well to modern transformer-based denoisers.

\begin{figure*}[htpb]
    \centering
    \subfloat[$2\times$SR; PnP-HQS+DRUNet]{\resizebox{0.32\linewidth}{!}{
    \definecolor{barcolor}{RGB}{100, 150, 255}
    \begin{tikzpicture}[y=0.75cm, x=1.25cm]
    \draw[thick] (0,6.2) -- (6.8,6.2);

    \node[anchor=west] at (0, 5.65) {\large\textbf{Variant}};
    \node[anchor=west] at (2.8, 5.65) {\large\textbf{Time (ms)}};

    \draw[thick] (0,5.1) -- (6.8,5.1);

    \node[anchor=west, text width=2.3cm, align=left] at (0.1, 4.5) {ViSTA};

    \filldraw[fill=barcolor, draw=none] (2.8, 4.3) rectangle +(2.6162811777833626, 0.4);
    \node[anchor=west] at (5.566281177783363, 4.5) {0.228};

    \node[anchor=west, text width=2.3cm, align=left] at (0.1, 3.7) {Equiv. (rand)};

    \filldraw[fill=barcolor, draw=none] (2.8, 3.5) rectangle +(1.689947600738769, 0.4);
    \node[anchor=west] at (4.639947600738769, 3.7) {0.109};

    \node[anchor=west, text width=2.3cm, align=left] at (0.1, 2.9000000000000004) {Equiv. (avgd)};

    \filldraw[fill=barcolor, draw=none] (2.8, 2.7) rectangle +(3.4000000000000004, 0.4);
    \node[anchor=west] at (6.3500000000000005, 2.9000000000000004) {0.389};

    \node[anchor=west, text width=2.3cm, align=left] at (0.1, 2.1000000000000005) {Vanilla};

    \filldraw[fill=barcolor, draw=none] (2.8, 1.9000000000000006) rectangle +(1.272977447899479, 0.4);
    \node[anchor=west] at (4.2229774478994795, 2.1000000000000005) {0.072};

    \draw[thick] (0,1.5000000000000004) -- (6.8,1.5000000000000004);

    \end{tikzpicture}
    }} \hfill
    \subfloat[Gaussian Deblur; PnP-PGD+DnCNN]{\resizebox{0.32\linewidth}{!}{
    \definecolor{barcolor}{RGB}{255,69,0}
    \begin{tikzpicture}[y=0.75cm, x=1.25cm]
    \draw[thick] (0,6.2) -- (6.8,6.2);

    \node[anchor=west] at (0, 5.65) {\large\textbf{Variant}};
    \node[anchor=west] at (2.8, 5.65) {\large\textbf{Time (ms)}};

    \draw[thick] (0,5.1) -- (6.8,5.1);

    \node[anchor=west, text width=2.3cm, align=left] at (0.1, 4.5) {ViSTA};

    \filldraw[fill=barcolor, draw=none] (2.8, 4.3) rectangle +(3.4000000000000004, 0.4);
    \node[anchor=west] at (6.3500000000000005, 4.5) {0.195};

    \node[anchor=west, text width=2.3cm, align=left] at (0.1, 3.7) {Equiv. (rand)};

    \filldraw[fill=barcolor, draw=none] (2.8, 3.5) rectangle +(1.7390098318515361, 0.4);
    \node[anchor=west] at (4.6890098318515365, 3.7) {0.068};

    \node[anchor=west, text width=2.3cm, align=left] at (0.1, 2.9000000000000004) {Equiv. (avgd)};

    \filldraw[fill=barcolor, draw=none] (2.8, 2.7) rectangle +(3.0787908119269254, 0.4);
    \node[anchor=west] at (6.028790811926926, 2.9000000000000004) {0.164};

    \node[anchor=west, text width=2.3cm, align=left] at (0.1, 2.1000000000000005) {Vanilla};

    \filldraw[fill=barcolor, draw=none] (2.8, 1.9000000000000006) rectangle +(2.0374881214532565, 0.4);
    \node[anchor=west] at (4.987488121453257, 2.1000000000000005) {0.085};

    \draw[thick] (0,1.5000000000000004) -- (6.8,1.5000000000000004);

    \end{tikzpicture}
    }} \hfill
    \subfloat[Motion Deblur; PnP-HQS+DiffUNet]{\resizebox{0.32\linewidth}{!}{
    \definecolor{barcolor}{RGB}{0,255,127}
    \begin{tikzpicture}[y=0.75cm, x=1.25cm]
    \draw[thick] (0,6.2) -- (6.8,6.2);

    \node[anchor=west] at (0, 5.65) {\large\textbf{Variant}};
    \node[anchor=west] at (2.8, 5.65) {\large\textbf{Time (ms)}};

    \draw[thick] (0,5.1) -- (6.8,5.1);

    \node[anchor=west, text width=2.3cm, align=left] at (0.1, 4.5) {ViSTA};

    \filldraw[fill=barcolor, draw=none] (2.8, 4.3) rectangle +(3.187960793965834, 0.4);
    \node[anchor=west] at (6.137960793965834, 4.5) {0.413};

    \node[anchor=west, text width=2.3cm, align=left] at (0.1, 3.7) {Equiv. (rand)};

    \filldraw[fill=barcolor, draw=none] (2.8, 3.5) rectangle +(2.1730899292460744, 0.4);
    \node[anchor=west] at (5.123089929246074, 3.7) {0.229};

    \node[anchor=west, text width=2.3cm, align=left] at (0.1, 2.9000000000000004) {Equiv. (avgd)};

    \filldraw[fill=barcolor, draw=none] (2.8, 2.7) rectangle +(3.4000000000000004, 0.4);
    \node[anchor=west] at (6.3500000000000005, 2.9000000000000004) {0.461};

    \node[anchor=west, text width=2.3cm, align=left] at (0.1, 2.1000000000000005) {Vanilla};

    \filldraw[fill=barcolor, draw=none] (2.8, 1.9000000000000006) rectangle +(1.960299844982764, 0.4);
    \node[anchor=west] at (4.910299844982764, 2.1000000000000005) {0.198};

    \draw[thick] (0,1.5000000000000004) -- (6.8,1.5000000000000004);

    \end{tikzpicture}
    }} \hfill
    \caption{Per-iteration runtime (in milliseconds; log scale) for three PnP variants—Vanilla-PnP, Equivariant-PnP, and ViSTA-PnP—combined with different frameworks: (a) PnP-HQS+DRUNet for $2\times$ superresolution, (b) PnP-PGD+DnCNN for Gaussian deblurring, and (c) PnP-HQS+DiffUNet for motion deblurring. Results are averaged over the CBSD10 dataset.
    }
    \label{fig:time_comp}
\end{figure*}
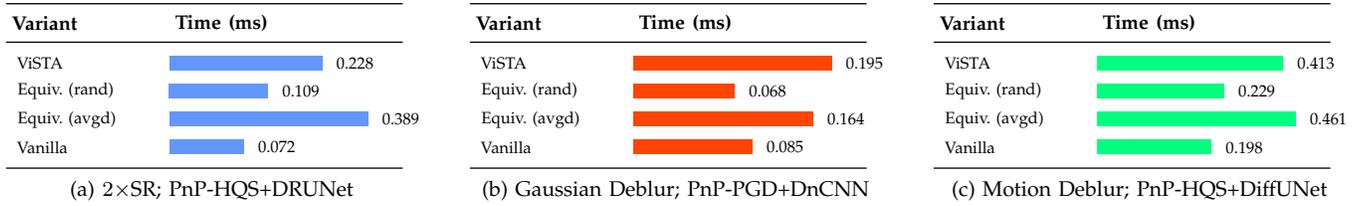

\subsection{Computational Aspects}
\label{sec:overheads}

The ViSTA-PnP algorithm in~\cref{algo:vista} includes two additional steps compared to Vanilla-PnP.  The first is computing the fixed point $\p$ of the viscosity operator $\cS$, which is done only once at the start of the reconstruction process. This takes about 2 seconds for a $256\times 256$ RGB image and adds minimal overhead. The second step is computing the viscosity index $\theta_k$ before applying $\cS$ in each iteration. This accounts for most of the extra cost introduced by ViSTA-PnP. In~\cref{fig:time_comp}, we show the per-iteration runtime overhead of ViSTA-PnP compared to Vanilla-PnP and Equivariant-PnP. The results show that ViSTA-PnP adds only a modest overhead per iteration, remaining faster than the more expensive $D_4$-averaged Equivariant-PnP variant.

On average across tasks, ViSTA-PnP is approximately $151\%$ slower than Vanilla-PnP, $106\%$ slower than Equivariant-PnP (random), but $23\%$ faster than Equivariant-PnP ($D_4$-averaged), reflecting a favorable trade-off between stability and runtime. In summary, ViSTA offers a good balance between stability and efficiency. It keeps the runtime practical while avoiding the high computational cost of equivariant averaging.

\subsection{Discussion}

We have shown that ViSTA is a robust and flexible framework for achieving stable and predictable PnP reconstructions with a wide range of pretrained denoisers.  Intuitively, viscosity helps by suppressing small artifacts early in the iterations, which, if left unchecked, could grow into poor reconstructions. The viscosity operator fills in these unstable regions, allowing the denoiser and the overall algorithm to stabilize and converge to a good solution, as seen in~\cref{fig:equi-diverge}, or to continue improving, as in~\cref{fig:mmo-psnr,fig:diff-psnr,fig:dru-psnr}. 

A key question is whether we can identify the causes of instability in PnP systems, like those illustrated in~\cref{fig:frontimage}. While having such explanations would allow for more targeted solutions, they are often hard to obtain and can depend heavily on the specific denoiser or proximal algorithm being used. This is where ViSTA stands out: it operates at the level of the overall operator and does not rely on understanding the exact source of instability. In other words, ViSTA can effectively stabilize a PnP system even when the underlying reason for its instability is unknown.

Although increasing viscosity may lead to a slight drop in PSNR, our experiments show that the peak performance is usually preserved or only mildly affected. Finally, it is important to note that early stopping is not a reliable fix for instability. As seen in results like~\cref{fig:mmo-psnr,fig:diff-psnr,fig:dru-psnr}, the point at which PSNR peaks varies across settings and is difficult to predict ahead of time.

\section{Conclusion}
\label{sec:conclusion}

We introduced ViSTA-PnP, a simple and effective baseline for stabilizing plug-and-play (PnP) methods with pretrained denoisers. We demonstrated its versatility across a range of base algorithms and denoising models. The stability guarantees offered by ViSTA-PnP complement the strong empirical performance of existing PnP frameworks. While ViSTA-PnP may lead to a modest reduction in PSNR in some cases, the primary objective of this work is stabilization rather than performance enhancement.

Two key components of ViSTA-PnP are the upper bound $\Theta$ and the viscosity operator defined in~\eqref{eq:NLM}. While $\Theta$ is simple to set, it is critical in balancing stability and reconstruction quality. The proposed viscosity operator typically outperforms the naive contraction $\cS = \beta I$ with $\beta \approx 1$, which often fails to stabilize the iterations effectively while maintaining high PSNR. However, this improvement comes at the cost of increased computational complexity.

Our findings point to several promising directions for future research, such as developing cost-efficient viscosity operators and exploring more advanced stabilization methods based on viscosity theory~\cite{attouch_viscosity_1996}. Although our experiments focused on deblurring and superresolution, the ViSTA-PnP framework can be applied to other inverse problems where PnP methods have proven successful~\cite{ahmad2020plug,yuan2020plug,jin2017deep}.

\bibliographystyle{IEEEtran}
\bibliography{reference}

\end{document}